%% file: main.tex
\documentclass[runningheads,a4paper]{llncs}  
\usepackage{times}
\usepackage{etex}
\usepackage{wrapfig}
\usepackage{fancyhdr}

\input{macros}

\begin{document}

\mainmatter

\title{Fragments of ML Decidable by\\
Nested Data Class Memory Automata}

\titlerunning{Fragments of ML Decidable by Nested Data Class Memory Automata}

\author{
Conrad Cotton-Barratt\inst{1}\fnmsep\thanks{Supported by an EPSRC Doctoral Training Grant}
\and David Hopkins\inst{1}\fnmsep\thanks{Supported by Microsoft Research and Tony Hoare. Now at Ensoft Limited, UK.}
\and Andrzej S. Murawski\inst{2}\fnmsep\thanks{Supported by EPSRC (EP/J019577/1)} 
\and C.-H. Luke Ong\inst{1},\fnmsep\thanks{Partially supported by Merton College Research Fund}
}

\institute{Department of Computer Science, University of Oxford, UK 
\and Department of Computer Science, University of Warwick, UK
}

\tocauthor{Conrad Cotton-Barratt (University of Oxford),
David Hopkins (University of Oxford),
Andrzej Murawski (University of Warwick),
Luke Ong (University of Oxford)}

\maketitle

\thispagestyle{fancy}

\begin{abstract}
The call-by-value language $\rml$ may be viewed as a canonical restriction of Standard ML to ground-type references, augmented by a ``bad variable'' construct in the sense of Reynolds.
We consider the fragment of (finitary) $\rml$ terms of order at most 1 with free variables of order at most 2, and identify two subfragments of this for which we show observational equivalence to be decidable.  The first subfragment, $\rmlpstr$, consists of those terms in which the P-pointers in the game semantic representation are determined by the underlying sequence of moves.  The second subfragment consists of terms in which the O-pointers of moves corresponding to free variables in the game semantic representation are determined by the underlying moves.
These results are shown using a reduction to a form of automata over data words in which the data values have a tree-structure, reflecting the tree-structure of the threads in the game semantic plays.
In addition we show that observational equivalence is undecidable at {every} third- or higher-order type, {every} second-order type which takes at least two first-order arguments, and every second-order type (of arity greater than one) that has a first-order argument which is not the final argument.
\end{abstract}

\section{Introduction}

{$\rml$} is a call-by-value functional language with state \cite{AM97b}.  It is similar to Reduced ML \cite{PS98}, the canonical restriction of Standard ML to ground-type references, except that it includes a ``bad variable'' constructor (in the absence of the constructor, the equality test is definable). This paper concerns the decidability of observational equivalence of finitary RML, $\frml$. 
Our ultimate goal is to classify the decidable fragments of $\frml$ completely. 
In the case of finitary Idealized Algol (IA), the decidability of observational equivalence depends only on the type-theoretic order \cite{MOW05} of the type sequents. In contrast, the decidability of $\frml$ sequents is not so neatly characterised by order (see Figure~\ref{tab:SummaryOfDecidabilityResultsForRML}): there are undecidable sequents of order as low as 2 \cite{Mur05}, amidst interesting classes of decidable sequents at each of orders 1 to 4. 

Following Ghica and McCusker \cite{GhicaM03}, we use game semantics to decide observational equivalence of $\frml$. Take a sequent $\seq{\Gamma}{M : \theta}$ with $\Gamma = x_{1} :\theta_{1}, \cdots, x_{n}:\theta_{n}$. In game semantics \cite{HY99}\cite{HylandO00}, the type sequent is interpreted as a P-strategy $\sem{\seq{\Gamma}{M : \theta}}$ for playing (against O, who takes the environment's perspective) in the prearena $\sem{\seq{\overline \theta}{\theta}}$. A play between P and O is a sequence of moves in which each non-initial move has a justification pointer to some earlier move -- its justifier. Thanks to the fully abstract game semantics of $\rml$, observational equivalence is characterised by \emph{complete plays} i.e.~$\seq{\Gamma}{M \cong N}$ iff the P-strategies, $\sem{\seq{\Gamma}{M}}$ and $\sem{\seq{\Gamma}{N}},$ contain the same set of {complete plays}. 
Strategies may be viewed as highly constrained processes, and are amenable to automata-theoretic representations; the chief technical challenge lies in the encoding of pointers. 

In \cite{HMO11} we introduced the \emph{O-strict} fragment of $\frml$, $\rmlostr$, consisting of sequents $\seq{x_{1}:\theta_{1}, \cdots, x_{n}:\theta_{n}}{M : \theta}$ such that $\theta$ is \emph{short} (i.e.~order at most 2 and arity at most 1), and every argument type of every $\theta_{i}$ is short.
Plays over prearenas denoted by O-strict sequents enjoy the property that the pointers from O-moves are uniquely determined by the underlying move sequence. The main result in \cite{HMO11} is that the set of complete plays of a $\rmlostr$-sequent is representable as a visibly pushdown automaton (VPA). A key idea is that it suffices to require each word of the representing VPA to encode the pointer from only \emph{one} P-question. The point is that, when the full word language is analysed, it will be possible to uniquely place all justification pointers.

The simplest type that is not O-strict is $\beta \rightarrow \beta \rightarrow \beta$ where $\beta \in \makeset{\intt,\unitt}$. Encoding the pointers from O-moves is much harder 
because O-moves are controlled by the environment rather than the term. As observational equivalence is defined by a quantification over all contexts, the strategy for a term must consider \emph{all} legal locations of pointer from an O-move, rather than just a single location in the case of pointer from a P-move. In this paper, we show that automata over data words can precisely capture strategies over a class of non-O-strict types. 

\paragraph{Contributions.} 
We identify two fragments of $\frml$ in which we can use deterministic weak nested data class memory automata \cite{Cotton-BarrattMO14} 
(equivalent to the locally prefix-closed nested data automata in \cite{Decker14}) to represent the set of complete plays of terms in these fragments.  
These automata operate over a data set which has a tree structure, and we use this structured data to encode O-pointers in words.

Both fragments are contained with the fragment $\forml$, which consists of terms-in-context $\seq{\Gamma}{M}$ where every type in $\Gamma$ is order at most 2, and the type of $M$ is order at most 1.  The first fragment, the \emph{P-Strict subfragment}, consists of those terms in $\forml$ for which in the game semantic arenas have the property that the P-pointers in plays are uniquely determined by the underlying sequence of moves.  
This consists of terms-in-context $\seq{\Gamma}{M: \theta}$ in which $\theta$ is any first order type, and each type in $\Gamma$ has arity at most $1$ and order at most 2. 
The second fragment, $\rforml$, consists of terms-in-context $\seq{\Gamma}{M: \theta}$ in which $\theta$, again, is any first order type, and each type $\theta' \in \Gamma$ is at most order $2$, such that each argument for $\theta'$ has arity at most 1.
Although these two fragments are very similar, they use different encodings of data values, and we discuss the difficulties in extending these techniques to larger fragments of $\frml$.

Finally we show that observational equivalence is undecidable at {every} third- or higher-order type, {every} second-order type which takes at least two first-order arguments, and every second-order type (of arity greater than one) that has a first-order argument which is not the final argument. See Figure~\ref{tab:SummaryOfDecidabilityResultsForRML} for a summary.

\begin{figure}[h]
	\centering
	\begin{tabular}{|c|c|c|c|}
	\hline
	Fragment & Representative Type Sequent & Recursion
	& Ref.\\
	\hline
	\hline
	\multicolumn{4}{|c|}{\textbf{Decidable}}\\
	\hline
	
			\begin{tabular}{c}O-Strict / $\rmlostr$ \\(\textsc{ExpTime}-Complete)\end{tabular} &
	 		\begin{tabular}{l}$\seq{((\beta \rightarrow \ldots \rightarrow \beta)\rightarrow \beta)\rightarrow \ldots  \rightarrow \beta}{}$\\$\qquad\qquad{(\beta\rightarrow \ldots \rightarrow \beta) \rightarrow \beta}$ \end{tabular} 
	 		& $\rmlterm{while}$ & \cite{Hopkins12,HMO11}\\
	
	\hline
	
			\begin{tabular}{c}O-Strict + Recursion \\(\textsc{DPDA}-Hard)\end{tabular} &
	 		\begin{tabular}{l}$\seq{((\beta \rightarrow \ldots \rightarrow \beta)\rightarrow \beta)\rightarrow \ldots  \rightarrow \beta}{}$\\$\qquad\qquad{(\beta\rightarrow \ldots \rightarrow \beta) \rightarrow \beta}$\end{tabular} 
	 		& $\beta \rightarrow \beta$ & \cite{Hopkins12} \\

	\hline
	
		$\rmlpstr$ &
		\begin{tabular}{c}$\seq{(\beta \rightarrow \dots \rightarrow \beta)\rightarrow \beta}{\beta\rightarrow  \dots \rightarrow \beta}$\end{tabular}&
		$\rmlterm{while}$ & \dag \\
	\hline
	
		$\rforml$ &
		\begin{tabular}{c}$\seq{(\beta \rightarrow \beta) \rightarrow \dots \rightarrow (\beta \rightarrow \beta) \rightarrow \beta}{}$\\${\beta \rightarrow \dots \rightarrow \beta}$\end{tabular}&
		$\rmlterm{while}$ & \dag \\
		\hline
	
	\hline
	\hline

	\multicolumn{4}{|c|}{\textbf{Undecidable}}\\
	
	\hline
	
			Third-Order &
			\begin{tabular}{c}$\seq{}{((\beta \rightarrow \beta) \rightarrow \beta)\rightarrow \beta}$\\$\seq{(((\beta \rightarrow \beta) \rightarrow \beta)\rightarrow \beta)\rightarrow \beta}{\beta}$\end{tabular} 
			& $\bot$ & \cite{Hopkins12},\dag \\
	
	\hline
	
			Second-Order &
			\begin{tabular}{c}$\seq{}{(\beta \rightarrow \beta) \rightarrow \beta\rightarrow \beta}$\\$\seq{((\beta \rightarrow \beta) \rightarrow \beta\rightarrow \beta) \rightarrow \beta}{\beta}$\end{tabular}
			& $\bot$ & \cite{Hopkins12},\dag \\
	
	\hline
	
			Recursion 
			& Any
			& $(\beta \rightarrow \beta) \rightarrow \beta$ & \cite{Hopkins12},\dag \\
	
	\hline
	\hline

	\multicolumn{4}{|c|}{\textbf{Unknown}}\\
	
	\hline
		$\forml$ & 
		\begin{tabular}{c} $(\beta \rightarrow \dots \rightarrow \beta) \rightarrow \dots \rightarrow (\beta \rightarrow \dots \rightarrow \beta)$\\$\seq{\rightarrow \beta}{\beta \rightarrow \dots \rightarrow \beta}$  \end{tabular} 
		& $\bot$ & - \\
		\hline
		$\hbox{\rm RML}_X$ & 
		\begin{tabular}{c} $\seq{}{\beta \rightarrow (\beta \rightarrow \beta) \rightarrow \beta}$\\$\seq{((\beta \rightarrow \beta)\rightarrow \beta)\rightarrow \beta}{\beta \rightarrow \beta \rightarrow \beta}$  \end{tabular} 
		& $\bot$ & - \\
		\hline
		$\hbox{FO RML}$ + Recursion & 
		$\seq{}{\beta\rightarrow  \dots \rightarrow \beta}$
		& ${\beta\rightarrow  \beta \rightarrow \beta}$ & - \\
		\hline
	\end{tabular}
	\caption{\label{tab:SummaryOfDecidabilityResultsForRML} Summary of RML Decidability Results. \small{($\dag$ marks new results presented here; $\beta \in \makeset{\intt,\unitt}$; we write $\bot$ to mean an undecidability result holds (or none is known) even if no recursion or loops are present, and the only source of non-termination is through the constant $\Omega$)}}
\vspace{-20pt}
\end{figure}

\paragraph*{Related Work.}
A related language with full ground references (i.e. with a ${\rmltype{int\ ref\ ref}}$ type) was studied in \cite{MurawskiT12}, and observational equivalence was shown to be undecidable even at types $\seq{}{\unitt \rightarrow \unitt \rightarrow \unitt}$. 
In contrast, for $\frml$ terms, we show decidability at the same type. The key technical innovation of our work is the use of automata over infinite alphabets to encode justification pointers. 
Automata over infinite alphabets have already featured in papers on game semantics~\cite{MT11,MurawskiT12} but there they were used for a different purpose, namely, to model fresh-name generation. 
The nested data class memory automata we use in this paper are an alternative presentation of locally prefix-closed data automata \cite{Decker14}.

\section{Preliminaries}
\subsubsection*{\hbox{RML}}
We assume base types $\unitt$, for commands, $\intt$ for a finite set of integers, and a integer variable type, $\intreft$.  Types are built from these in the usual way.
The \emph{order} of a type $\theta \rightarrow \theta'$ is given by $max( order(\theta) + 1, order(\theta'))$, where base types $\unitt$ and $\intt$ have order $0$, and $\intreft$ has order 1.  The \emph{arity} of a type $\theta \rightarrow \theta'$ is $arity(\theta') + 1$ where $\unitt$ and $\intt$ have arity $0$, and $\intreft$ has arity 1.
A full syntax and set of typing rules for $\rml$ is given in Figure \ref{fig:rml-syntax}.
Note though we include only the arithmetic operations $\suc{i}$ and $\pre{i}$, these are sufficient to define all the usual comparisons and operations.  We will write $\letin{x=M}{N}$ as syntactic sugar for $(\lambda x. N) M$, and $M;N$ for $(\lambda x. N) M$ where $x$ is a fresh variable.

\begin{figure}
\vspace{-20pt}
\begin{mathpar}
\inferrule{ }{ \seq{\Gamma}{\mathsf{()}:\unitt}}

\inferrule{ i\in\natnum}{ \seq{\Gamma}{i:\intt}}

\inferrule{\seq{\Gamma}{M : \intt}}{\seq{\Gamma}{\suc{M} : \intt}}

\inferrule{\seq{\Gamma}{M : \intt}}{\seq{\Gamma}{\pre{M} : \intt}}

\inferrule{ \seq{\Gamma}{M:\intt} \\
          \seq{\Gamma}{M_0:\theta} \\ 
          \seq{\Gamma}{M_1:\theta}}{ 
          \seq{\Gamma}{\cond{M}{M_1}{M_0}:\theta}}
          
\inferrule{\seq{\Gamma}{M:\intreft}}{\seq{\Gamma}{!M:\intt}}

\inferrule{ \seq{\Gamma}{M:\intreft} \\ \seq{\Gamma}{N:\intt}}{ \seq{\Gamma}{\assg{M}{N}:\unitt}}

\inferrule{\seq{\Gamma}{M : \intt}}{ \seq{\Gamma}{{\refint{M}} : \intreft}}

\inferrule{ }{ \seq{\Gamma,x:\theta}{x:\theta}}

\inferrule{\seq{\Gamma}{M:\theta\rightarrow\theta'} \\ \seq{\Gamma}{N:\theta}}{ \seq{\Gamma}{MN:\theta'}}

\inferrule{\seq{\Gamma,x:\theta}{M:\theta'}}{ \seq{\Gamma}{\lambda x^\theta.M:\theta\rightarrow\theta'}}

\inferrule{ \seq{\Gamma}{M:\intt} \\ \seq{\Gamma}{N:\unitt}}{ \seq{\Gamma}{\while{M}{N}:\unitt}}

\inferrule{ \seq{\Gamma}{M:\unitt\rightarrow\intt}\\ \seq{\Gamma}{N:\intt\rightarrow\unitt}}{\seq{\Gamma}{\makevar{M}{N}:\intreft} }

\end{mathpar}
\vspace{-10pt}
\caption{Syntax of RML}\label{fig:rml-syntax}
\vspace{-10pt}
\end{figure}

The operational semantics, defined in terms of a big-step relation, are standard \cite{Mur05}.
For closed terms $\vdash M$ we write $M{\Downarrow}$ just if there exist $s, V$ such that $\emptyset, M \Downarrow s,V$.  Two terms $\seq{\Gamma}{M : \theta}$ and $\seq{\Gamma}{N:\theta}$ are \emph{observationally equivalent} (or \emph{contextually equivalent}) if for all (closing) contexts $C[-]$ such that $\seq{\emptyset}{C[M], C[N] : \unitt}$, $C[M]{\Downarrow}$ if and only if $C[N]{\Downarrow}$. 

It can be shown that every $\rml$ term is effectively convertible to an equivalent term in \emph{canonical form} \cite[Prop.~3.3]{Hopkins12}, defined by the following grammar ($\beta \in \makeset{\unitt,\intt}$).
\[\small
\begin{array}{ll}
\can & ::= \, ()  \, | \, i \, | \, x^\beta \, | \, \suc{x^\beta} \, | \, \pre{x^\beta} \,|\, \cond{x^\beta}{\can}{\can} \,|\, x^\intreft\aasg y^\intt \,|\, !x^\intreft \,|\\ 
&  \lambda x^\theta.\can \,| \, \makevar{\lambda x^\unitt.\can}{\lambda y^\intt.\can} \,|\,  \letin{x=\newc}{\can} \,|\, \while{\can}{\can}  \,|\, \letin{x^\beta=\can}{\can} \,| \\
&  \letin{x= z y^\beta}{\can} \,|\,\
\letin{x= z\, \makevar{\lambda u^\unitt.\can}{\lambda v^\intt.\can}}{\can} \,|\, \letin{x=z(\lambda x^\theta.\can)}{\can}
\end{array}\]

\subsubsection*{Game Semantics}
We use a presentation of call-by-value game semantics in the style of Honda and Yoshida \cite{HY99}, as opposed to Abramsky and McCusker's isomorphic model \cite{AM97b}, as Honda and Yoshida's more concrete constructions lend themselves more easily to recognition by automata.  We recall the following presentation of the game semantics for $\rml$ from \cite{HMO11}.

An \emph{arena} $A$ is a triple $(M_{A}, \vdash_{A}, \lambda_{A})$ where $M_A$ is a set of \emph{moves} where $I_A \subseteq M_A$ consists of \emph{initial} moves, $\vdash_A \subseteq M_A \times (M_A \backslash I_A)$ is called the \emph{justification relation}, and $\lambda_A : M_A \rightarrow \makeset{O,P}\times \makeset{Q,A}$ a labelling function such that for all $i_A \in I_A$ we have $\lambda_A(i_A) = (P,A)$ and if $m \vdash_A m'$ then $(\pi_1 \lambda_A)(m) \neq (\pi_1 \lambda_A)(m')$ and $(\pi_2 \lambda_A)( m') = A \Rightarrow (\pi_2 \lambda_A)(m) = Q$.

The function $\lambda_A$ labels moves as belonging to either \emph{Opponent} or \emph{Proponent} and as being either a \emph{Question} or an \emph{Answer}.  Note that answers are always justified by questions, but questions can be justified by either a question or an answer.  We will use arenas to model types.  However, the actual games will be played over \emph{prearenas}, which are defined in the same way except that initial moves are O-questions.  

Three basic arenas are $0$, the empty arena, $1$, the arena containing a single initial move $\bullet$, and $\intnum$, which has the integers as its set of moves, all of which are initial P-answers.    
The constructions on arenas are defined in Figure~\ref{fig:arenaConstructions}.  
Here we use $\overline{I_A}$ as an abbreviation for $M_A \backslash I_A$, and $\overline{\lambda_A}$ for the O/P-complement of $\lambda_A$.   Intuitively $A\otimes B$ is the union of the arenas $A$ and $B$, but with the initial moves combined pairwise.  $A \Rightarrow B$ is slightly more complex.  First we add a new initial move, $\bullet$.  We take the O/P-complement of $A$, change the initial moves into questions, and set them to now be justified by $\bullet$.  Finally, we take $B$ and set its initial moves to be justified by $A$'s initial moves.  The final construction, 
$A \rightarrow B$, takes two arenas $A$ and $B$ and produces a prearena, as shown below.  This is essentially the same as $A \Rightarrow B$ without the initial move $\bullet$.

\begin{figure}
\vspace{-20pt}
\small	
\[\begin{array}{rclcrcl}
	M_{A\Rightarrow B} & = & \makeset{\bullet} \uplus M_A \uplus M_B&& M_{A\otimes B} & = & I_A \times I_B \uplus \overline{I_A} \uplus \overline{I_B}\\
	I_{A\Rightarrow B} & = & \makeset{\bullet} && I_{A\otimes B} & = & I_A \times I_B \\
	\lambda_{A\Rightarrow B} & = & m \mapsto 
		\left\{ \begin{array}{ll} 
			PA &\quad \mbox{if $m = \bullet$} \\ 
			OQ &\quad \mbox {if $m \in I_A$} \\ 
			\overline{\lambda_A}(m) &\quad \mbox{if $m \in \overline{I_A}$} \\ 
			\lambda_B(m) &\quad \mbox{if $m \in M_B$} \end{array} 
		\right.  &&
		\lambda_{A\otimes B} & = & m \mapsto 
		\left\{ \begin{array}{ll} 
			PA &\quad \mbox{if $m \in I_A \times I_B$} \\
			\lambda_A(m) &\quad \mbox{if $m \in \overline{I_A}$} \\ 
			\lambda_B(m) &\quad \mbox{if $m \in \overline{I_B}$} \end{array} 
		\right.  \\
	\vdash_{A\Rightarrow B} & = & \makeset{(\bullet, i_A) | i_A \in I_A} && 
	\vdash_{A\otimes B} & = & \{ ((i_A,i_B),m) | i_A \in I_A \wedge i_B \in I_B \\
	                        &   & \cup \makeset{(i_A,i_B) | i_A \in I_A, i_B \in I_B} &&&& \quad \wedge (i_A \vdash_A m \vee i_B \vdash_B m)\}\\	
	                        &   & \cup \vdash_A \cup \vdash_B	&&&&\cup (\vdash_A \cap (\overline{I_A} \times \overline{I_A})) \\
	                        &&&&&&\cup (\vdash_B \cap (\overline{I_B} \times \overline{I_B}))\\
	\end{array}\]
\[\begin{array}{ll}
	M_{A\rightarrow B} \; = \; M_A \uplus M_B \qquad &
	\lambda_{A\rightarrow B}(m) \; = \; 
		\left\{ \begin{array}{ll} 
			OQ &\quad \mbox {if $m \in I_A$} \\ 
			\overline{\lambda_A}(m) &\quad \mbox{if $m \in \overline{I_A}$} \\ 
			\lambda_B(m) &\quad \mbox{if $m \in M_B$} \end{array} 
		\right.  
\\	
	I_{A\rightarrow B} \; = \; I_A 
&
	\vdash_{A\rightarrow B} \; = \; \makeset{(i_A,i_B) | i_A \in I_A, i_B \in I_B} \cup \vdash_A \cup \vdash_B	\\
	\end{array}\]	
	\caption{Constructions on Arenas 
	\label{fig:arenaConstructions}}
	\vspace{-15pt}
\end{figure}

We intend arenas to represent types, in particular $\sem{\unitt} = 1$, $\sem{\intt}  = \intnum$ (or a finite subset of $\intnum$ for $\frml$) and $\sem{\theta_1 \rightarrow \theta_2} = \sem{\theta_1} \Rightarrow \sem{\theta_2}$.  A term $x_1 : \theta_1, \ldots, x_n : \theta_n \vdash M : \theta$ will be represented by a \emph{strategy} for the prearena $\sem{\theta_1} \otimes \ldots \otimes \sem{\theta_n} \rightarrow \sem{\theta}$.

A \emph{justified sequence} in a prearena $A$ is a sequence of moves from $A$ in which the first move is initial and all other moves $m$ are equipped with a pointer to an earlier move $m'$, such that $m' \vdash_A m$. A \emph{play} $s$ is a justified sequence which additionally satisfies 
the standard conditions of Alternation, Well-Bracketing, and Visibility.

A \emph{strategy} $\sigma$ for prearena $A$ is a non-empty, even-prefix-closed set of plays from $A$, satisfying the determinism condition: if $s\, m_1, s\, m_2 \in \sigma$ then $s\, m_1 = s\, m_2$.
We can think of a strategy as being a playbook telling P how to respond by mapping odd-length plays to moves. 
A play is \emph{complete} if all questions have been answered.  Note that (unlike in the call-by-name case) a complete play is not necessarily maximal.  We denote the set of complete plays in strategy $\sigma$ by $\comp{\sigma}$.

In the game model of $\rml$, a term-in-context $x_1 : \theta_1, \ldots, x_n : \theta_n \vdash M : \theta$ is interpreted by a strategy of the prearena $\sem{\theta_1} \otimes \ldots \otimes \sem{\theta_n} \rightarrow \sem{\theta}$.  These strategies are defined by recursion over the syntax of the term.
Free identifiers $\seq{x : \theta}{x : \theta}$ are interpreted as \emph{copy-cat} strategies where P always copies O's move into the other copy of $\sem{\theta}$, $\lambda x. M$ allows multiple copies of $\sem{M}$ to be run, application $M N$ requires a form of parallel composition plus hiding and the other constructions can be interpreted using special strategies.  The game semantic model is fully abstract in the following sense.

\begin{theorem}[Abramsky and McCusker \cite{AM96,AM97b}]
If $\seq{\Gamma}{M~:~\theta}$ and $\seq{\Gamma}{N~:~\theta}$ are $\rml$ type sequents, then $\seq{\Gamma}{M \cong N}$ iff $\comp{\sem{\seq{\Gamma}{M}}} =  \comp{\sem{\seq{\Gamma}{N}}}$.
\end{theorem}

\subsubsection*{Nested Data Class Memory Automata}

We will be using automata to recognise game semantic strategies as languages.  Equality of strategies can then be reduced to equivalence of the corresponding automata.  However, to represent strategies as languages we must encode pointers in the words.  
To do this we use data languages, in which every position in a word has an associated \emph{data value}, which is drawn from an infinite set (which we call the \emph{data set}).  Pointers between positions in a play can thus be encoded in the word by the relevant positions having suitably related data values.
Reflecting the hierarchical structure of the game semantic prearenas, we use a data set with a tree-structure.

Recall a \emph{tree} is a simple directed graph  $\anglebra{D,pred}$ where $pred: D \rightharpoonup D$ is the predecessor map defined on every node of the tree except the root, such that every node has a unique path to the root.  A node $n$ has level $l$ just if $pred^l(n)$ is the root (thus the root has level 0).  A tree is of level $l$ just if every node in it has level $\leq l$.  We define a \emph{nested data set} of level $l$ to be a tree of level $l$ such that each data value of level strictly less than $l$ has infinitely many children.
We fix a nested data set of level $l$, $\dataset$, and a finite alphabet $\Sigma$, to give a data alphabet $\D = \Sigma \times \dataset$.

We will use a form of automaton over these data sets based on class memory automata
\cite{BjorklundS10}.  Class memory automata operate over an unstructured data set, and on reading an input letter $(a,d)$, the transitions available depend both on the state the automaton is currently in, and the state the automaton was in after it last read an input letter with data value $d$.  We will be extending a weaker variant of these automata, in which the only acceptance condition is reaching an accepting state.
The variant of class memory automata we will be using, nested data class memory automata \cite{Cotton-BarrattMO14}, works similarly: on reading input $(a,d)$ the transitions available depend on the current state of the automaton, the state the automaton was in when it last read a descendant (under the $\pred$ function) of $d$, and the states the automaton was in when it last read a descendant of each of $d$'s ancestors.  We also add some syntactic sugar (not presented in \cite{Cotton-BarrattMO14}) to this formalism, allowing each transition to determine the automaton's memory of where it last saw the read data value and each of its ancestors: this does not extend the power of the automaton, but will make the constructions we make in this paper easier to define.

Formally, a Weak Nested Data Class Memory Automaton (WNDCMA) of level $l$ is a tuple $\anglebra{Q, \Sigma, \Delta, q_0, F}$ where $Q$ is the set of states, $q_0 \in Q$ is the initial state, $F \subseteq Q$ is the set of accepting states, and the transition function $\delta = \bigcup_{i=0}^{l} \delta_i$ where each $\delta_i$ is a function:
$$ \delta_i : Q \times \Sigma \times (\set{i} \times (Q \uplus \set{\bot})^{i+1}) \rightarrow \powerset (Q \times Q^{i+1}) $$
We write $Q_\bot$ for the set $Q \uplus \set{\bot}$, and may refer to the $Q_\bot^{j}$ part of a transition as its \emph{signature}.  The automaton is \emph{deterministic} if each set in the image of $\delta$ is a singleton.
A configuration is a pair $(q,f)$ where $q \in Q$, and $f: \dataset \rightarrow Q_\bot$ is a class memory function (i.e. $f(d) = \bot$ for all but finitely many $d \in \dataset$).  The initial configuration is $(q_0, f_0)$ where $f_0$ is the class memory function mapping every data value to $\bot$.  
The automaton can transition from configuration $(q,f)$ to configuration $(q',f')$ on reading input $(a,d)$ just if $d$ is of level-$i$, $(q',(t_0, t_1, \dots ,t_i)) \in \delta ( q, a, (i, f(pred^{i} (d), \dots, f(pred(d)), f(d)))$, and $f' = f[d \mapsto t_i, pred(d) \mapsto t_{i-1}, \dots, pred^{i-1}(d) \mapsto t_1, pred^i(d) \mapsto t_0]$.
A run is defined in the usual way, and is accepting if the last configuration $(q_n, f_n)$ in the run is such that $q_n \in F$. We say $w \in L(\calA)$ if there is an accepting run of $\calA$ on $w$.

Weak nested data class memory automata have a decidable emptiness problem, reducible to coverability in a well-structured transition system \cite{Cotton-BarrattMO14,Decker14}, and are closed under union and intersection by the standard automata product constructions.  Further, Deterministic WNDCMA are closed under complementation again by the standard method of complementing the final states.  Hence they have a decidable equivalence problem.

\section{P-Strict $\forml$}
In \cite{HMO11}, the authors identify a fragment of RML, the O-strict fragment, for which the plays in the game-semantic strategies representing terms have the property that the justification pointers of $O$-moves are uniquely reconstructible from the underlying moves.
Analogously, we define the P-strict fragment of RML to consist of typed terms in which the pointers for $P$-moves are uniquely determined by the underlying sequence of moves.  Then our encoding of strategies for this fragment will only need to encode the $O$-pointers: for which we will use data values.

\subsection{Characterising P-Strict $\rml$}

In working out which type sequents for $\rml$ lead to prearenas which are P-strict, it is natural to ask for a general characterisation of such prearenas.  The following lemma, which provides exactly that, is straightforward to prove:
\begin{lemma}
A prearena is P-strict iff there is no enabling sequence $q \vdash \dots \vdash q'$ in which both $q$ and $q'$ are P-questions.
\end{lemma}
Which type sequents lead to a P-question hereditarily justifying another P-question?    It is clear, from the construction of the prearena from the type sequent, that if a free variable in the sequent has arity $>1$ or order $>2$, the resulting prearena will have a such an enabling sequence, so not be P-strict.  Conversely, if a free variable is of a type of order at most $2$ and arity at most $1$, it will not break P-strictness.
On the RHS of the type sequent, things are a little more complex: there will be a ``first'' P-question whenever the type has an argument of order $\geq 1$.  To prevent this P-question hereditarily justifying another P-question, the argument must be of arity $1$ and order $\leq 2$.  Hence the P-strict fragment consists of type sequents of the following form:
$$
\seq{(\beta \rightarrow \dots \rightarrow \beta) \rightarrow \beta}{ ((\beta \rightarrow \dots \rightarrow \beta) \rightarrow \beta) \rightarrow \dots \rightarrow ((\beta \rightarrow \dots \rightarrow \beta) \rightarrow \beta) \rightarrow \beta }
$$
(where $\beta \in \set{\unitt, \intt}$.)

From results shown here and in \cite{Hopkins12}, we know that observational equivalence of all type sequents with an order 3 type or order 2 type with order 1 non-final argument on the RHS are undecidable.  Hence the only P-strict types for which observational equivalence may be decidable are of the form:
$
\seq{(\beta \rightarrow \dots \rightarrow \beta) \rightarrow \beta}{\beta \rightarrow \dots \rightarrow \beta}
$
or
$
\seq{(\beta \rightarrow \dots \rightarrow \beta) \rightarrow \beta}{\beta \rightarrow \dots \rightarrow \beta \rightarrow (\beta \rightarrow \beta) \rightarrow \beta}
$.
In this section we show that the first of these, which is the intersection of the P-strict fragment and $\forml$, does lead to decidability.

\begin{definition}
The P-Strict fragment of $\forml$, which we denote $\rmlpstr$, consists of typed terms of the form
$
\seq{x_1 : \widehat{\Theta_1}, \dots, x_n : \widehat{\Theta_1}}{M : \Theta_1}
$
where the type classes $\Theta_i$ are as described below:
$$\Theta_0 ::= \unitt \: | \: \intt  
\quad\quad
\Theta_1 ::= \Theta_0 \: | \: \Theta_0 \rightarrow \Theta_1 \: | \: \intreft  \quad\quad
\widehat{\Theta_1} ::= \Theta_0 \: | \: \Theta_1 \rightarrow \Theta_0 \: | \: \intreft 
$$
\end{definition}
This means we allow types of the form
$
\seq{(\beta \rightarrow \dots \rightarrow \beta) \rightarrow \beta}{\beta \rightarrow \dots \rightarrow \beta}
$
where $\beta \in \set{\unitt, \intt}$.

\subsection{Deciding Observational Equivalence of $\rmlpstr$}
Our aim is to decide observational equivalence by constructing, from a term $M$, an automaton that recognises a language representing $\sem{M}$.   As $\sem{M}$ is a set of plays, the language representing $\sem{M}$ must encode both the moves and the pointers in the play.  Since answer moves' pointers are always determined by well-bracketing, we only represent the pointers of question moves, and we do this with the nested data values.  The idea is simple: if a play $s$ is in $\sem{M}$ the language $L(\sem{M})$ will contain a word, $w$, such that the string projection of $w$ is the underlying sequence of moves of $s$, and such that:
\begin{itemize}
\item The initial move takes the (unique) level-0 data value; and
\item Answer moves take the same data value as that of the question they are answering; and
\item Other question moves take a fresh data value whose predecessor is the data value taken by the justifying move.
\end{itemize}
Of course, the languages recognised by nested data automata are closed under automorphisms of the data set, so in fact 
each play $s$ will be represented by an infinite set of data words, all equivalent to one another by automorphism of the data set.

\begin{theorem}\label{thm:pstrict-proof}
For every typed term $\seq{\Gamma}{M:\theta}$ in $\rmlpstr$ that is in canonical form we can effectively construct a deterministic weak nested data class memory automata, $\calA^M$, recognising the complete plays of $L(\sseq{\Gamma}{M})$.
\end{theorem}
\begin{proof}
We prove this by induction over the canonical forms.  We note that for each canonical form construction, if the construction is in $\rmlpstr$ then each constituent canonical form must also be.  
For convenience of the inductive constructions, we in fact construct automata $\calA^M_\gamma$ recognising $\sseq{\Gamma}{M}$ restricted to the initial move $\gamma$.  
Here we sketch two illustrative cases.
A full proof is provided in Appendix \ref{appendix:pstrict}.

\textbf{\boldmath{$\lambda x^\beta . M : \beta \rightarrow \theta$}.}
The prearenas for $\sem{M}$ and $\sem{\lambda x^\beta. M}$ are shown in Figure \ref{fig:prearenas-lambda}.  Note that in this case we must have that $\seq{\Gamma, x: \beta}{M:\theta}$, and so the initial moves in $\sem{M}$ contain an $x$-component.  We therefore write these initial moves as $(\gamma, i_x)$ where $\gamma$ is the $\Gamma$-component and $i_x$ is the $x$-component.

\begin{figure}
\vspace{-10pt}
\begin{subfigure}{0.48\textwidth}
	\[\begin{tikzpicture}[arena]
		\node[clear](q1){$q_1$};
		\node[clear](a1)[below = of q1]{$a_1$} edge (q1);
		\node[clear](dots) [below = of a1] {$\vdots$} edge (a1);
		\node[clear](ql)[below = of dots]{$q_{n}$} edge (dots);
		\node[clear](al)[below = of ql]{$a_{n}$} edge (ql);
		\node[tri, shape border uses incircle,shape border rotate=60,scale = 0.7](gamma) at (q1.220){$\sem{\Gamma}$};	
	\end{tikzpicture}\]
 	\subcaption{$\sem{\seq{\Gamma, \beta}{\theta}}$}
\end{subfigure}
\begin{subfigure}{0.48\textwidth}
	\[\begin{tikzpicture}[arena]
		\node[clear](q0){$q_0$};
		\node[clear](a0)[below = of q0]{$a_0$} edge (q0);
		\node[clear](q1)[below = of a0]{$q_1$} edge (a0);
		\node[clear](a1)[below = of q1]{$a_1$} edge (q1);
		\node[clear](dots) [below = of a1] {$\vdots$} edge (a1);
		\node[clear](ql)[below = of dots]{$q_{n}$} edge (dots);
		\node[clear](al)[below = of ql]{$a_{n}$} edge (ql);
		\node[tri, shape border uses incircle,shape border rotate=60,scale = 0.7](gamma) at (q0.220){$\sem{\Gamma}$};	
	\end{tikzpicture}\]
	\subcaption{$\sem{\seq{\Gamma}{\beta \rightarrow \theta}}$}
\end{subfigure}
\caption{Prearenas for $\sem{\seq{\Gamma, x:\beta}{M:\theta}}$ and $\sem{\seq{\Gamma}{\lambda x^\beta . M : \beta \rightarrow \theta}}$}\label{fig:prearenas-lambda}
\vspace{-15pt}
\end{figure}

P's strategy $\sem{\lambda x^\beta . M}$ is as follows: after an initial move ${\gamma}$, P plays the unique $a_0$-move $\bullet$, and waits for a $q_1$-move.  Once O plays a $q_1$-move $i_x$, P plays as in $\sem{\seq{\Gamma,x}{M}}$ when given an initial move $({\gamma}, i_x)$.
However, as the $q_1$-moves are not initial, it is possible that O will play another $q_1$-move, $i_x'$.  Each time O does this it opens a new thread which P plays as per $\sem{\seq{\Gamma,x}{M}}$ when given initial move $({\gamma}, i_x')$.  Only O may switch between threads, and this can only happen immediately after P plays an $a_j$-move (for any $j$).  

By our inductive hypothesis, for each initial move $(\gamma, i_x)$ of $\sem{\seq{\Gamma, x : \beta}{\theta}}$ we have an automaton $\calA^M_{\gamma,i_x}$ recognising the complete plays of $\sem{\seq{\Gamma, x: \beta}{M:\theta}}$ starting with the initial move $(\gamma, i_x)$.  
We construct the automaton $\calA^{\lambda x . M}_\gamma$ by taking a copy of each $\calA^M_{\gamma,i_x}$, and quotient together the initial states of these automata to one state, $p$, (which by conditions on the constituent automata we can assume has no incoming transitions).  This state $p$ will hold the unique level-0 data value for the run, and states and transitions are added to have initial transitions labelled with $q_0$ and $a_0$, ending in state $p$.  The final states will be the new initial state, the quotient state $p$, and the states which are final in the constituent automata.
The transitions inside the constituent automata fall into two categories: those labelled with moves corresponding to the RHS of the term in context $\seq{\Gamma}{M}$, and those labelled with moves corresponding to the LHS.  Those transitions corresponding to moves on the RHS are altered to have their level increased by 1, with their signature correspondingly altered by requiring a level-0 data value in state $p$.  Those transitions corresponding to moves on the LHS retain the same level, but have the top value of their data value signature replaced with the state $p$.  Finally, transitions are added between the constituent automata to allow switching between threads: whenever there is a transition out of a final state in one of the automata, copies of the transition are added from every final state (though keeping the data-value signature the same).  
Note that the final states correspond to precisely the points in the run where the environment is able to switch threads.

\textbf{\boldmath{$\letin{x^\beta = M}{N}$}.}
Here we assume we have automata recognising $\sem{M}$ and $\sem{N}$.  The strategy $\sem{\letin{x^\beta = M}{N}}$ essentially consists of a concatenation of $\sem{M}$ and $\sem{N}$, with the result of playing $\sem{M}$ determining the value of $x$ to use in $\sem{N}$.  Hence the automata construction is very similar to the standard finite automata construction for concatenation of languages, though branching on the different results for $\sem{M}$ to different automata for $\sem{N}$.
\end{proof}

\begin{corollary}
Observational equivalence of terms in $\rmlpstr$ is decidable
\end{corollary}

\section{A Restricted Fragment of $\forml$}
It is important, for the reduction to nested data automata for $\rmlpstr$, that variables cannot be partially evaluated: in prearenas where variables have only one argument, once a variable is evaluated those moves cannot be used to justify any future moves.  If we could later return to them we would need ensure that they were accessed only in ways which did not break visibility.  We now show that this can be done, using a slightly different encoding of pointers, for a fragment in which variables have unlimited arity, but each argument for the variable must be evaluated all at once.  This means that the variables have their O-moves uniquely determined by the underlying sequence of moves.

\subsection{Fragment definition}

\begin{definition}
The fragment we consider in this section, which we denote $\rforml$, consists of typed terms of the form
$
\seq{x_1: \Theta_2^1, \dots, x_n : \Theta_2^1}{M:\Theta_1}
$
where the type classes $\Theta_i$ are as described below:
\begin{align*}
&\Theta_0 ::= \unitt \: | \: \intt  &\quad &\Theta_1^1 ::= \Theta_0 \: | \: \Theta_0 \rightarrow \Theta_0 \: | \: \intreft \\
&\Theta_1 ::= \Theta_0 \: | \: \Theta_0 \rightarrow \Theta_1 \: | \: \intreft &\quad &\Theta_2^1 ::= \Theta_1 \: | \: \Theta_1^1 \rightarrow \Theta_2^1 
\end{align*}
\end{definition}

\begin{wrapfigure}{r}{0.53\textwidth}
\vspace{-20pt}
	\[\begin{tikzpicture}[arena]		
		\node[clear](q0){$q_0$};
		
		\node[clear](vdots)[below =of q0] {$\vdots$} edge(q0);
		\node[clear](hdots1)[left =of vdots] {$\dots$};

		\node[clear](a0)[right = 2em of vdots]{$a_0$} edge(q0);
		\node[clear](q1)[below =of a0]{$q_1$} edge(a0);
		\node[clear](a1)[below =of q1]{$a_1$} edge(q1);
		\node[clear](vdots1)[below =of a1]{$\vdots$} edge(a1);
		\node[clear](qn)[below =of vdots1]{$q_n$} edge(vdots1);
		\node[clear](an)[below =of qn]{$a_n$} edge(qn);

		\node[clear](qt1)[left =of hdots1]{$q^{(1)}$} edge(q0);
		\node[clear](q0t1)[below left =of qt1] {$q^{(1)}_0$} edge(qt1);
		\node[clear](a0t1)[below left =of q0t1] {$a^{(1)}_0$} edge(q0t1);
		\node[clear](at1)[below =of qt1] {$a^{(1)}$} edge(qt1);
		\node[clear](qt2)[below =of at1]{$q^{(2)}$} edge(at1);
		\node[clear](q0t2)[below left =of qt2] {$q^{(2)}_0$} edge(qt2);
		\node[clear](a0t2)[below left =of q0t2] {$a^{(2)}_0$} edge(q0t2);
		\node[clear](at2)[below =of qt2] {$a^{(2)}$} edge(qt2);
		\node[clear](vdots2)[below =of at2]{$\vdots$} edge(at2);
		\node[clear](qtk)[below =of vdots2]{$q^{(k)}$} edge(vdots2);
		\node[clear](q0tk)[below left =of qtk] {$q^{(k)}_0$} edge(qtk);
		\node[clear](a0tk)[below left =of q0tk] {$a^{(k)}_0$} edge(q0tk);
		\node[clear](atk)[below =of qtk] {$a^{(k)}$} edge(qtk);
		
		\node[draw,dotted,fit=(a0) (q1) (a1) (vdots1) (qn) (an),label=A] {};
		\node[draw,dotted,fit=(qt1) (at1) (qt2) (at2) (vdots2) (qtk) (atk),label=B] {};
		\node[draw,dotted,fit=(q0t1) (a0t1) (q0t2) (a0t2) (q0tk) (a0tk) ,label=C] {};

		\end{tikzpicture}\]
		\vspace{-10pt}
\caption{Shape of arenas in $\rforml$}\label{fig:1frorml-prearenas}
\vspace{-30pt}
\end{wrapfigure}
This allows types of the form
$
\seq{(\beta \rightarrow \beta) \rightarrow \dots \rightarrow (\beta \rightarrow \beta) \rightarrow \beta}{\beta \rightarrow \dots \rightarrow \beta}
$
where $\beta \in \set{\unitt,\intt}$.  The shape of the prearenas for this fragment is shown in Figure \ref{fig:1frorml-prearenas}.  Note that moves in section $A$ of the prearena (marked in Figure \ref{fig:1frorml-prearenas}) relate to the type $\Theta_1$ on the RHS of the typing judgement, and that we need only represent O-pointers for this section, since the P-moves are all answers so have their pointers uniquely determined by well-bracketing.  Moves in sections $B$ and $C$ of the prearena correspond to the types on the LHS of the typing judgement.  Moves in section $B$ need only have their P-pointers represented, since the O-moves are all answer moves.  Moves in section $C$ have both their O- and P-pointers represented by the underlying sequence of moves: the P-pointers because all P-moves in this section are answer moves, the O-pointers by the visibility condition.

\subsection{Deciding Observation Equivalence}

Similarly to the P-Strict case, we provide a reduction to weak nested data class memory automata that uses data values to encode O-pointers.  However, this time we do not need to represent any O-pointers on the LHS of the typing judgement, so use data values only to represent pointers of the questions on the RHS.  
We do, though, need to represent P-pointers of moves on the LHS.  This we do using the same technique used for representing P-pointers in \cite{HMO11}: in each word in the language we represent only one pointer by using a ``tagging" of moves: the string $s \psource{m} s' \ptarget{m'}$ is used to represent the pointer $\Pstr{(s) {s} \ (m) {m} \ (s') {s'} \ (m'-m) {m'}}$.  Because P's strategy is deterministic, representing one pointer in each word is enough to uniquely reconstruct all P-pointers in the plays from the entire language.  
Due to space constraints we do not provide a full explanation of this technique in this paper: for a detailed discussion see \cite{Hopkins12,HMO11}.  Hence for a term $\sseq{\Gamma}{M:\theta}$ the data language we seek to recognise, $L(\sseq{\Gamma}{M})$ represents pointers in the following manner:
\begin{itemize}
\item The initial move takes the (unique) level-0 data value;
\item Moves in $\sem{\Gamma}$ (i.e. in section $B$ or $C$ of the prearena) take the data value of the previous move;
\item Answer moves in $\sem{\theta}$ (i.e. in section $A$ of the prearena) take the data value of the question they are answering; and
\item Non-initial question moves in $\sem{\theta}$ (i.e. in section $A$ of the prearena) take a fresh data value nested under the data value of the justifying answer move.
\end{itemize}

\begin{theorem}\label{thm:forml-proof}
For every typed term $\seq{\Gamma}{M:\theta}$ in $\rforml$ that is in canonical form we can effectively construct a deterministic weak nested data class memory automaton, $\calA_M$, recognising the complete plays of $L(\sseq{\Gamma}{M})$.
\end{theorem}
\begin{proof}
This proof takes a similar form to that of Theorem \ref{thm:pstrict-proof}: by induction over canonical forms. We here sketch the $\lambda$-abstraction case.  A full proof is provided in Appendix \ref{appendix:rforml}.

\textbf{\boldmath{$\lambda x^\beta . M : \beta \rightarrow \theta$}.}  
This construction is almost identical to that in the proof of Theorem \ref{thm:pstrict-proof}: again the strategy for P is interleavings of P's strategy for $M : \theta$.  The only difference in the construction is that where in the encoding for Theorem \ref{thm:pstrict-proof} the moves in each $\calA^M_{\gamma,i_x}$ corresponding to the LHS and RHS of the prearena needed to be treated separately, in this case they can be treated identically: all being nested under the new level-0 data value.  We demonstrate this construction in Example \ref{ex:automata-lambda-construction}
\end{proof}

\begin{example}\label{ex:automata-lambda-construction}
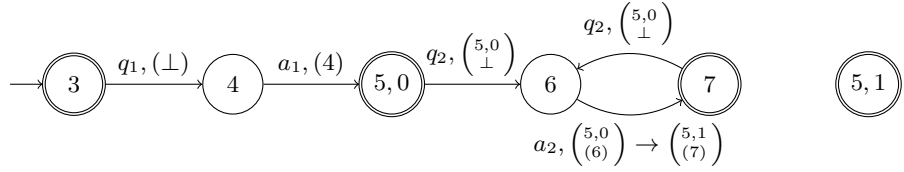
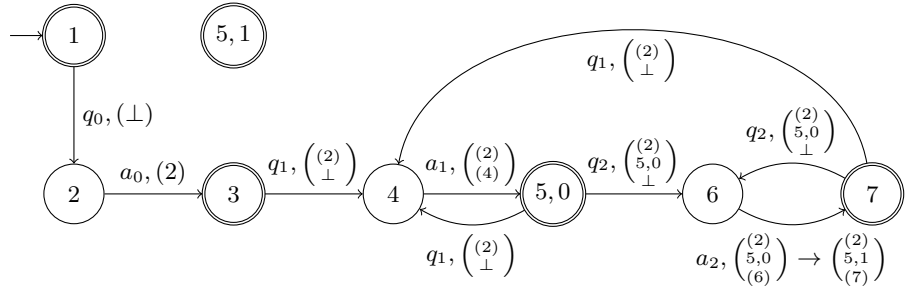
\begin{figure}
\vspace{-20pt}
\begin{subfigure}[b]{0.98\textwidth}
\begin{tikzpicture}[->,baseline,auto,node distance=2.09cm,font=\small]
  \node[state,initial,accepting] (3) {$3$};  
  \node[state] (4) [right of=3] {$4$};
  \node[state,accepting] (5) [right of=4] {$5,0$};
  \node[state](6) [right of=5] {$6$};
  \node[state,accepting](7) [right of=6] {$7$};
  \node[state,accepting] (5') [right of=7] {$5,1$};

  \path (3) edge              node {$q_1,(\bot)$} 		(4)
   (4) edge              node {$a_1,(4)$} 		(5)
   (5) edge              node {$q_2,\twovec{5,0}{\bot}$} 		(6)
   (6) edge  [bend right,below]            node {$a_2,\twovec{5,0}{(6)} \rightarrow \twovec{5,1}{(7)} $} 		(7)
   (7) edge  [bend right,above]           node {$q_2,\twovec{5,0}{\bot}$} 		(6);
\end{tikzpicture}
\caption{Automaton for $\sseq{}{\letin{c = \newc}{\lambda y^\unitt . \cond{\deref c = 0}{c := 1}{\Omega} }}$}\label{subfig:automata-letin}
\end{subfigure}

\begin{subfigure}[b]{1\textwidth}
\begin{tikzpicture}[->,baseline,auto,node distance=2.1cm,font=\small]
  \node[state,initial,accepting] (1) {$1$};  
  \node[state] (2) [below of=1] {$2$};
  \node[state,accepting] (3) [right of=2] {$3$};  
  \node[state] (4) [right of=3] {$4$};
  \node[state,accepting] (5) [right of=4] {$5,0$};
  \node[state,accepting] (5') [above of=3] {$5,1$};
  \node[state](6) [right of=5] {$6$};
  \node[state,accepting](7) [right of=6] {$7$};
  \path (1) edge              node {$q_0,(\bot)$} 		(2)
   (2) edge              node {$a_0,(2)$} 		(3)
   (3) edge              node {$q_1,\twovec{(2)}{\bot}$} 		(4)
   (4) edge              node {$a_1,\twovec{(2)}{(4)}$} 		(5)
   (5) edge              node {$q_2,\threevec{(2)}{5,0}{\bot}$} 		(6)
   (6) edge  [bend right,below]            node {$a_2,\threevec{(2)}{5,0}{(6)} \rightarrow \threevec{(2)}{5,1}{(7)} $} (7)
   (7) edge  [bend right,above]           node {$q_2,\threevec{(2)}{5,0}{\bot}$} 		(6)
   (5) edge  [bend left]            node {$q_1,\twovec{(2)}{\bot}$} 		(4)
   (7) edge  [out=100,in=80]            node {$q_1,\twovec{(2)}{\bot}$} 		(4)
   ;
\end{tikzpicture}
\caption{Automaton for $\sseq{}{\lambda x^\unitt . \letin{c = \newc}{\lambda y^\unitt . \cond{\deref c = 0}{c := 1}{\Omega} }}$}\label{subfig:automata-lambda}
\end{subfigure}
\caption{Automata recognising strategies}\label{fig:automata-example}
\vspace{-10pt}
\end{figure}

Figure \ref{fig:automata-example} shows two weak nested data class memory automata.  We draw a transition $p, a, (j, \threevec{s_0}{\vdots}{s_j}) \rightarrow p', \threevec{s_0'}{\vdots}{s_j'} \in \delta$ as an arrow from state $p$ to $p'$ labelled with ``$a, \threevec{s_0}{\vdots}{s_j} \rightarrow \threevec{s_0'}{\vdots}{s_j'}$''.  
We omit the ``$\rightarrow \threevec{s_0'}{\vdots}{s_j'}$'' part of the label if $s_j' = p'$ and $s_i = s_i'$ for all $i \in \set{0,1, \dots, j-1}$.

The automaton obtained by the constructions in Theorem \ref{thm:forml-proof} for the term-in-context $\sseq{}{\letin{c = \newc}{\lambda y^\unitt . \cond{\deref c = 0}{c := 1}{\Omega} }}$ is shown in Figure \ref{subfig:automata-letin} (to aid readability, we have removed most of the dead and unreachable states and transitions).  Note that we have the states $(5,0)$ and $(5,1)$ - here the second part of the state label is the value of the variable $c$: the top-level data value will remain in one of these two states, and by doing so store the value of $c$ at that point in the run.  The move $q_2$ in this example corresponds to the environment providing an argument $y$: note that in a run of the automaton the first time a $y$ argument is passed, the automaton proceeds to reach an accepting state, but in doing so sets the top level data value to the state $(5,1)$.  This means the outgoing transition shown from state $7$ cannot fire.

The automaton for $\sseq{}{\lambda x^\unitt . \letin{c = \newc}{\lambda y^\unitt . \cond{\deref c = 0}{c := 1}{\Omega} }}$ is shown in Figure \ref{subfig:automata-lambda} (again, cleaned of dead/unreachable transitions for clarity).  Note that this contains the first automaton as a sub-automaton, though with a new top-level data value added to the transitions. The $q_1$ move now corresponds to providing a new argument for $x$, thus starting a thread. Transitions have been added from the accepting states $(5)$ and $(7)$, allowing a new $x$-thread to be started from either of these locations.  
Note that the transition from $(7)$ to $(6)$, which could not fire before, now can fire because several data values (corresponding to different $x$-threads) can be generated and left in the state $(5,0)$.  
\end{example}

\section{Undecidable Fragments}
In this section we consider which type sequents and forms of recursion are expressive enough to prove undecidability.  The proofs of the results this section proceed by identifying terms such that the induced complete plays correspond to runs of Turing-complete machine models.  
Full proofs are given in Appendix \ref{appendix:undecidability}.

\paragraph{On the Right of the Turnstile.}
In~\cite{Mur03} it is shown that observational equivalence is undecidable for 5th-order terms.  The proof takes the strategy that was used to show undecidability for 4th-order IA and finds an equivalent call-by-value strategy.  
It is relatively straightforward to adapt the proof to show that observational equivalence is undecidable at 3rd-order types, e.g.\ $((\unitt \rightarrow \unitt) \rightarrow \unitt) \rightarrow \unitt$.
A further result in~\cite{Mur05} showed that the problem is undecidable at the type $(\unitt \rightarrow \unitt) \rightarrow (\unitt \rightarrow \unitt) \rightarrow \unitt$.
Both results easily generalise to show that the problem is undecidable at \emph{every} 3rd-order type and \emph{every} 2nd-order type which takes at least two 1st-order arguments. We modify the second of these proofs to show undecidability at $(\unitt \rightarrow \unitt) \rightarrow  \unitt \rightarrow \unitt$. 
Our proof of this easily adapts to a proof of the following.  

\begin{theorem}\label{thm:rhsundecidability}
Observational equivalence is undecidable at every 2nd-order type (of arity at least two) which contains a 1st-order argument that is not the final argument.
\end{theorem}

\paragraph{On the Left of the Turnstile.}
Note that $\seq{}{M \cong N : \theta}$ if, and only if, $\seq{f : \theta \rightarrow \unitt}{f M \cong f N : \unitt}$.  Thus, for any sequent $\seq{}{\theta}$ at which observational equivalence is undecidable, the sequent $\seq{\theta \rightarrow \unitt}{\unitt}$ is also undecidable. So the problem is undecidable if, on the left of the turnstile, we have a fourth-order type or a (third-order) type which has a second-order argument whose first-order argument is not the last.

\paragraph{Recursion.}
In IA, observational equivalence becomes undecidable if we add recursive first-order functions \cite{Ong04}. The analogous results for \rml\ with recursion also hold:
\begin{theorem}\label{thm:recursive-undecidability}
Observational equivalence is undecidable in $\rmlostr$ equipped with recursive functions $(\unitt \rightarrow \unitt) \rightarrow \unitt$
\end{theorem}

\section{Conclusion}

We have used two related encodings of pointers to data values to decide two related fragments of $\forml$: $\rmlpstr$, in which the free variables were limited to arity 1, and $\rforml$, in which the free variables were unlimited in arity but each argument of the free variable was limited to arity 1.  
It is natural to ask whether we can extend or combine these approaches to decide the whole of $\forml$.  
Here we discuss why this seems likely to be impossible with the current machinery used.

In deciding $\rmlpstr$ we used the nested data value tree-structure to mirror the shape of the prearenas.  These data values can be seen as names for different threads, with the sub-thread relation captured by the nested structure.  
What happens if we attempt to use this approach to recognise strategies on types where the free variables have arity greater than 1?  With free variables having arity 1, whenever they are interrogated by P, they are entirely evaluated immediately: they cannot be partially evaluated.  With arity greater than 1, this partial evaluation can happen: P may provide the first argument at some stage, and then at later points evaluate the variable possibly several times with different second arguments.  P will only do this subject to visibility conditions though: if P partially evaluates a variable $x$ while in a thread $T$, it can only continue that partial evaluation of $x$ in $T$ or a sub-thread of $T$.  
This leads to problems when our automata recognise interleavings of similar threads using the same part of the automaton.  If P's strategy for the thread $T$ is the strategy $\sem{M}$ for a term $M$, and recognised by an automaton $\calA^M$, then $\sem{\lambda y . M}$ will consist of interleavings of $\sem{M}$. 
The automaton $\calA^{\lambda y . M}$ will use a copy of $\calA^M$ to simulate an unbounded number of $M$-threads.  
If $T$ is one such thread, which performs a partial evaluation of $x$, this partial evaluation will be represented by input letters with data values unrelated to the data value of $T$.  
If a sibling of $T$, $T'$, does the same, the internal state of the automaton will have no way of telling which of these partial evaluations was performed by $T$ and which by $T'$.  
Hence it may recognise data words which represent plays that break the visibility condition.

Therefore, to recognise strategies for terms with free variables of arity greater than 1, 
the natural approach to take is to have the data value of free-variable moves be related to the thread we are in.  
This is the approach we took in deciding $\rforml$: the free variable moves precisely took the data value of the part of the thread they were in.  
Then information about the partial evaluation was stored by the thread's data value.  
This worked when the arguments to the free variables had arity at most 1: however if we allow the arity of this to increase we need to start representing O-pointers in the evaluation of these arguments.  
For this to be done in a way that makes an inductive construction work for $\letin{x = (\lambda y . M)}{N}$, we must use some kind of nesting of data values for the different $M$-threads.  The na\"{i}ve approach to take is to allow the $M$-thread data values to be nested under the data value of whatever part of the $N$-thread they are in.  However, the $M$-thread may be started and partially evaluated in one part of the $N$-thread, and then picked up and continued in a descendant part of that $N$-thread. The data values used in continuing the $M$-thread must therefore be related to the data values used to represent the partial evaluation of the $M$-thread, but also to the part of the $N$-thread the play is currently in.  This would break the tree-structure of the data values, and so seem to require a richer structure on the data values.

\paragraph{Further Work.}
A natural direction for further work, therefore, is to investigate richer data structures and automata models over them that may provide a way to decide $\forml$.  

The automata we used have a non-primitive recursive emptiness problem, and hence the resulting algorithms both have non-primitive recursive complexity also.  Although work in \cite{Hopkins12} shows that this is not the best possible result in the simplest cases, the exact complexities of the observational equivalence problems are still unknown.

To complete the classification of $\frml$ also requires deciding (or showing undecidable) the fragment containing order 2 types (on the RHS) with one order 1 argument, which is the last argument.  
A first step to deciding this would be the fragment labelled $\hbox{RML}_X$ in figure \ref{tab:SummaryOfDecidabilityResultsForRML}.  
Deciding this fragment via automata reductions similar to those in this paper would seem to require both data values to represent O-pointers, and some kind of visible stack to nest copies of the body of the function, as used in \cite{HMO11}.
In particular, recognising strategies of second-order terms such as $\lambda f. f()$ requires the ability to recognise data languages (roughly) of the form $\{ d_1 d_2 ... d_n d_n ... d_2 d_1 \: | \: n \in N,$ each $d_i$ is distinct$\}$.  A simple pumping argument shows such languages cannot be recognised by nested data class memory automata, and so some kind of additional stack would seem to be required.

\bibliographystyle{plain}
\bibliography{consolidatedbib}

\appendix
\section{Proof of Theorem \ref{thm:pstrict-proof}}\label{appendix:pstrict}
Given a $\rmlpstr$ term-in-context $\seq{\Gamma}{M}$ we construct a Deterministic Weak NDCMA $\calA_{\seq{\Gamma}{M}}$ recognising, as a language, $\comp{\sem{\seq{\Gamma}{M}}}$.  By the full abstraction theorem, observational equivalence can then be checked by testing the corresponding automata for equivalence.

For notational convenience, in this appendix we write $\twovec{a}{d}$ for the letter $(a,d) \in \D$.
\begin{wrapfigure}{r}{0.45\textwidth}
\[
\begin{tikzpicture}[arena]		
		\node[clear](q0){$q_0$};
		
		\node[clear](vdots)[below =of q0] {$\vdots$} edge(q0);
		\node[clear](hdots1)[left =of vdots] {$\dots$};

		\node[clear](a0)[right = 4em of vdots]{$a_0$} edge(q0);
		\node[clear](q1)[below =of a0]{$q_1$} edge(a0);
		\node[clear](a1)[below =of q1]{$a_1$} edge(q1);
		\node[clear](vdots1)[below =of a1]{$\vdots$} edge(a1);
		\node[clear](qn)[below =of vdots1]{$q_n$} edge(vdots1);
		\node[clear](an)[below =of qn]{$a_n$} edge(qn);

		\node[clear](qt1)[left =of hdots1]{$q'$} edge(q0);
		\node[clear](q0t1)[below left=of qt1] {$q_0'$} edge(qt1);
		\node[clear](a0t1)[below =of q0t1] {$a_0'$} edge(q0t1);
		\node[clear](dotst1)[below =of a0t1] {$\vdots$} edge(a0t1);
		\node[clear](qmt1)[below =of dotst1] {$q_{m}'$} edge(dotst1);
		\node[clear](amt1)[below =of qmt1] {$a_{m}'$} edge(qmt1);
		\node[clear](at1)[below =of qt1] {$a'$} edge(qt1);


		\end{tikzpicture}
		\]
	\caption{Shape of prearenas for $\rmlpstr$}
	\label{fig:prearena-pstrict}
\end{wrapfigure}
The shape of the pre-arena for terms $\sem{\seq{\Gamma}{M}}$ in $\rmlpstr$ is shown in figure \ref{fig:prearena-pstrict}.  The moves in on the right of the prearena correspond to $M$, while moves on the left correspond to $\Gamma$.

For type sequents $\seq{\Gamma}{\theta}$ in $\rmlpstr$, a play $p$ in $\sem{\seq{\Gamma}{\theta}}$ is represented in the data language as a word $w$ where string projection of $w$ is equal to the underlying sequence of moves in $p$.  Pointers are only ambiguous for question moves in sections $A$ and $C$ of the arena.  Pointers for questions are represented in the following manner:
\begin{itemize}
\item The initial question takes a (fresh) level-0 data value.
\item If $a$ is an answer-move in the play, then the corresponding letter in the word will be $\twovec{a}{d}$ where $d$ is the same data value as the answer's justifier.
\item Question moves in sections $A$, $B$, and $C$ of the arena take a fresh data value $d$, such that $pred(d)$ is the data value of the justifying move.
\end{itemize}
Essentially: question moves take a data value whose predecessor is the data value of the justifying move, answer moves take the data value of the question they answer.

We note that this has the following (convenient) consequence: each data class of a word in such a language is either empty or of the form $\twovec{q}{d} \twovec{a}{d}$.  It cannot be longer than this.

\textbf{Reduction from $\rmlpstr$.} \quad
The reduction is inductive on the construction of the canonical form.  We make the construction indexed by initial moves, with each automaton $\calA_i$ recognising the appropriate language restricted to the initial move $i$.  The construction to combine these into one automaton as per the specification above is a  straightforward union of the automata and merging of the initial states.

Our inductive hypothesis is slightly stronger than that the constructed automata recognises the appropriate languages.  We also require the following conditions on the automaton $\calA^M_i$:
\begin{itemize}
\item Initial states are never revisited (or have data values assigned to them)
\item The automaton is deterministic
\item Each state can only ever ``hold" data values of one, fixed, level.
\item There is precisely one transition from the initial state, labelled $i, (0,\bot)$.  We will call the target state of this transition the ``secondary state" of the automaton.  Further, this is the only transition in the automaton with signature $(0,\bot)$.
\item If $q$ and $q'$ are (non-initial) final states in the automaton, then if there is a transition $(q,a,\xi,p,\xi')$ then $(q',a,\xi,p,\xi')$ is also a transition.
\end{itemize} 

\textbf{Notation describing NDCMA.} \quad
In the following, I represent transitions of NDCMA in a couple of ways.  The most standard notation I use is to write something of the form $p \xrightarrow{m, (k, \bar{p})} q, \bar{q}$.  Here we have $m \in \Sigma$, $p,q \in Q$, and $\bar{p}, \bar{q} \in (Q_\bot)^k$.  This means that $(q, \bar{q}) \in \delta_k ( p, m, (k, \bar{p}))$.  

I may write $\twovec{s}{\bar{s}}$ for the $k+1$-vector of elements of $Q_\bot$ obtained by putting $s$ ``on top" of the $k$-vector $\bar{s}$.  Similarly $\twovec{\bar{s}}{s}$ puts $s$ ``below" $\bar{s}$. 

Sometimes I omit the final $\bar{q}$: in this case it is implicitly assumed to only update the currently-read data value, which is updated to $q$.  Formally: this means $\bar{q} = \bar{p}[q / p_k]$.

When I am omitting this final $\bar{q}$, it is also possible to draw the automata in a relatively standard manner (e.g. in the first couple of cases below).

\subsection{$(): \unitt$}
For $\sem{\seq{\Gamma}{(): \unitt}}$ the complete plays of the strategy are of the form $\Pstr{(gamma) {\gamma} \ (bullet-gamma) {\bullet}}$ (or the empty play).  Hence $\calA_i$ is simply:

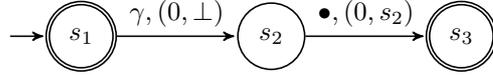
\begin{tikzpicture}[->,>=stealth',shorten >=1pt,auto,node distance=2.5cm,
                    semithick]

  \node[state,initial,accepting] (1) {$s_1$};  
  \node[state] (2) [right of=1] {$s_2$};
  \node[accepting,state] (3) [right of=2] {$s_3$};
  \path (1) edge              node {$\gamma, (0,\bot)$} 		(2)
   (2) edge              node {$\bullet, (0,s_2)$} 		(3);
\end{tikzpicture}

\subsection{$i: \intt$}
This is also straightforward, identical to the last case but with a differently labelled move:

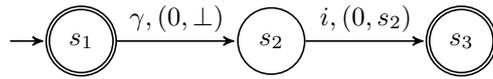
\begin{tikzpicture}[->,>=stealth',shorten >=1pt,auto,node distance=2.5cm,
                    semithick]

  \node[state,initial,accepting] (1) {$s_1$};  
  \node[state] (2) [right of=1] {$s_2$};
  \node[accepting,state] (3) [right of=2] {$s_3$};
  \path (1) edge              node {$\gamma, (0,\bot)$} 		(2)
   (2) edge              node {$i, (0,s_2)$} 		(3);
\end{tikzpicture}

\subsection{$x^\beta : \beta$}
Here we have $\seq{\Gamma}{x^\beta}$, so $x : \beta$ is in $\Gamma$.  Thus the initial moves have an $x$-component, so an initial move is of the form $(\bar{\gamma}, j)$ where $j$ is in the $x$-component.  For such an initial move, the plays recognised are just $\set{ \twovec{(\bar{\gamma},j)}{d} \twovec{j}{d} \: : \: d \in \dataset \text{ is level-0}}$, and again the appropriate automaton is straightforwardly given:

\begin{tikzpicture}[->,>=stealth',shorten >=1pt,auto,node distance=3cm,
                    semithick]

  \node[state,initial,accepting] (1) {$s_1$};  
  \node[state] (2) [right of=1] {$s_2$};
  \node[accepting,state] (3) [right of=2] {$s_3$};
  \path (1) edge              node {$(\bar{\gamma},j) , (0,\bot)$} 		(2)
   (2) edge              node {$j, (0,s_2)$} 		(3);
\end{tikzpicture}

\subsection{$\suc{x^\intt} : \intt$ and $\pre{x^\intt} : \intt$}
These are just as in the previous case, but adding or subtracting one to the $j$ (modulo the fragment of $\intnum$ being used).

\subsection{$x^{\intreft} := y^\intt : \unitt$}
Here we have $\seq{\Gamma}{x^{\intreft} := y^\intt}$, so $x : \intreft$ and $y: \intt$ are in $\Gamma$.  Thus the initial moves have a $y$-component, say $j$.  Thus the language recognised by $\calA_{(\bar{\gamma},j)}$ is just:
$$\set{ \twovec{(\bar{\gamma},j)}{d} \twovec{write_x (j)}{d'} \twovec{ok_x}{d'} \twovec{\bullet}{d} \: | \: d \text{ is level-0} \text{ and } pred(d') = d}$$
This is recognised by the following automaton:

\begin{tikzpicture}[->,>=stealth',shorten >=1pt,auto,node distance=3cm,
                    semithick,]
                   
  \node[state,initial,accepting] (1) {$s_1$};  
  \node[state] (2) [right of=1] {$s_2$};
  \node[state] (3) [right of=2] {$s_3$};
  \node[state] (4) [right of=3] {$s_4$};
  \node[accepting,state] (5) [right of=4] {$s_5$};
  \path (1) edge              node {\small $(\bar{\gamma},j) , (0,\bot)$} 		(2)
   (2) edge              node {\small $w_x(j), (1, \twovec{s_2}{\bot})$} 		(3)
   (3) edge              node {\small $ok_x, (1,\twovec{s_2}{s_3})$} 		(4)
   (4) edge              node {\small $\bullet, (0,s_2)$} 		(5);
\end{tikzpicture}

\subsection{$\deref{x^\intreft} : \intt$}
This is similar to the previous case, except that the value for P to return is given by O's reponse to $read_x$.  The desired language for $\calA_\gamma$ is is just 
$$\set{ \twovec{\gamma}{d} \twovec{read_x}{d'} \twovec{j_x}{d'} \twovec{j}{d} \: | \: j \in \natnum, d \text{ is level-0} \text{ and } pred(d') = d}$$
This is recognised by a very similar automaton to the previous case, except that from state $s_3$ the automaton splits into different states for each possible answer $j_x$.

\subsection{$\cond{x^\beta}{M}{N} : \theta$}
The initial move contains an $x$-component.  If this $x$-component is $0$ then the automaton is as the as the automaton for $N$, otherwise it is as the automaton for $M$.

\subsection{$\makevar{\lambda x^\unitt . M}{\lambda y^\intt . N} : \intreft$}
This construction is very similar to that provided in the $\rforml$ case, in appendix \ref{proof:1rforml-mkvar}.  The only difference is that now the automata for $M$ and $N$ needn't be level-0, as they may make plays in $\Gamma$.  These parts of the automata must retain the data-levels used, but nested under the initial level-0 data value used: this is a simple construction.  

\subsection{$\while{M}{N}$}
The strategy $\sem{\while{M}{N}}$ plays as if playing $M$ until the final move would be made.  If this would be 0, P gives the $\bullet$ answer to the initial move, and stops.  Otherwise it plays as if playing $N$, until the final move would be made, when it starts as if playing $M$ again.  This is easy to construct from the automata $\calA_\gamma^M$ and $\calA_\gamma^N$: for a full formal description of how, see appendix \ref{proof:1rforml-while} which deals with the $\rforml$ case.  The construction here is very similar: the only difference is that the constituent automata may no longer be level-0.  This could lead to difficulties if data values spawned in one run-throughs of the loop could be used in a later run-through, but this cannot happen as arguments cannot be partially evaluated in this fragment, so cannot be returned to later.

\subsection{$\letin{x = \newc}{M} : \theta$}
We assume we have a family of automata, $\calA^M_i$, recognising the strategy $\sem{\seq{\Gamma, x: \intreft}{M: \theta}}$. 
$\sem{\seq{\Gamma}{\letin{x = \newc}{M} : \theta}}$ is constructed by restricting behaviour of $x$ to ``good variable" behaviour (i.e. after a read-move the response is an immediate reply of the last integer written to the variable), and then hiding those moves.  The automata construction is done in these two stages.

\textbf{Restriction to good-variable behaviour.} \quad 
The value of the variable will be stored in both the current state of the automaton, and by the level-0 data value.  The level-0 data value will be used to ensure that when O switches between threads (see the $\lambda$-abstraction construction in section \ref{proof:rmlpstr-lambda}), the correct variable value is retained.  By keeping the value in the current state, the correct value is retained when moves in $\Gamma$ are being made.
Assume the finitary fragment we are using is $\set{0,1, \dots, K}$.  
Let $Q$ be the non-initial states of $\calA^M_\gamma$.
We construct $\calC_\gamma$ as follows:
\begin{itemize}
\item The states of the automaton are $\set{q_I} \uplus (Q \times \set{0,1, \dots, K})$
\item The final states are $q_I$ and those which are final in $\calA^M_\gamma$ paired with any integer.
\item The initial state is $q_I$
\item The transitions are given as follows:
\begin{itemize}
	\item $q_I \xrightarrow{q_0, (0,\bot)} (s_M,0)$ where $s_M$ is the secondary state of $\calA^M_\gamma$.
	\item if $q_1 \xrightarrow{m, (k, \threevec{s_0}{\vdots}{s_k})} q_2, \threevec{t_0}{\vdots}{t_k}$ is in $\calA^M_\gamma$	where $m$ is not an $x$-$write$ move or a response to an $x$-$read$ move, then:
	\begin{itemize}
		\item if $m$ is a $q_j$ move for some $j \neq 0$, for each $i_0, i_1, \dots, i_k \in \set{0,1, \dots, K}$, we have the transition $(q_1,i_0) \xrightarrow{m, (k, \threevec{(s_0,i_0)}{\vdots}{(s_k,i_k)})} (q_2,i_0), \threevec{(t_0,i_0)}{\vdots}{(t_k,i_0)}$ (for convenience, if $s_k = \bot$ we interpret $(s_k,i)$ as $\bot$ also).
		\item if $m$ is not a $q_j$ move, for each $i, j_0, j_1, \dots, j_k \in \set{0,1,\dots,K}$, we have the transition $(q_1,i) \xrightarrow{m, (k, \threevec{(s_0,j_0)}{\vdots}{(s_k,j_k)})} (q_2,i), \threevec{(t_0,i)}{\vdots}{(t_k,i)}$.
	\end{itemize}
	
	\item For each $j$, if $q_1 \xrightarrow{write_x(j), (1, \twovec{s_0}{\bot})} q_2, \twovec{t_0}{t_1}$ is in $\calA^M_\gamma$, then we have the transition $(q_1,i_1) \xrightarrow{write_x(j), (1, \twovec{(s_0,i_0)}{\bot})} (q_2,j), \twovec{(t_0,j)}{(t_1,j)}$ for each $i_1, i_0$.
	
	\item For each response to an $x$-read move, $j_x$, if $q_1 \xrightarrow{j_x, (1, \twovec{s_0}{s_1})} q_2, \twovec{t_0}{t_1}$ is in $\calA^M_\gamma$, then we have $(q_1,j) \xrightarrow{j_x, (1, \twovec{(s_0,i_0)}{(s_1,i_1)})} (q_2,j), \twovec{(t_0,j)}{(t_1,j)}$ for each $i_0, i_1$.
\end{itemize}
\end{itemize}
\noindent Note that adding the transitions between accepting states as required by the inductive hypothesis will not change the language recognised, since all the outward transitions added would be labelled with a $q_j$ move, and by construction these moves require a level-0 data value to be in the correct place.

\textbf{Hiding} \quad 
$\calA^{\letin{x = \newc}{M} }_\gamma$ is constructed from $\calC_\gamma$ as follows:

If we are in a configuration $(s_1, f)$ of $\calC_i$ where we can perform a transition $s_1 \xrightarrow{m_x, (j,\bar{s})} s_2, \bar{t}$ where $m_x$ is an $x$-move then by determinacy of strategies combined with the restriction to good variable behaviour, it is the only possible transition from this configuration.  
Thus for every state $s_0$ of $\calC_\gamma$ and every possible ``signature" $\twovec{t_0}{t_0'}$, there is a unique maximal (and not necessarily finite) sequence of transitions:
$$ s_0 \xrightarrow{m_0, (1, \twovec{t_0}{\bot})} s_1, \twovec{t_1}{t_1'} \xrightarrow{m_1, (1, \twovec{t_1}{t_1'})} s_2, \twovec{t_2}{t_2'} \xrightarrow{m_2, (1, \twovec{t_2}{\bot})} \dots $$
or
$$ s_0 \xrightarrow{m_0, (1, \twovec{t_0}{t_0'})} s_1, \twovec{t_1}{t_1'} \xrightarrow{m_1, (1, \twovec{t_1}{\bot})} s_2, \twovec{t_2}{t_2'} \xrightarrow{m_2, (1, \twovec{t_2}{t_2'})} \dots $$
where each $m_i$ is an $x$-move.

From each $\calC_\gamma$ we construct the automaton $\calA^{\letin{x = \newc}{M}}_\gamma$ by considering where this sequence terminates for each state.  Everything is the same as in $\calC_\gamma$ except for the transition relation, which is altered as follows:
\begin{itemize}
\item If the maximal sequence of $x$-moves with signature $\twovec{t_0}{t_0'}$ out of state $s_0$ is empty then all transitions requiring signature $\twovec{t_0}{t_0'}$ out of $s_0$ are unchanged.
\item If the maximal sequence out of $s_0$ with signature $\twovec{t_0}{t_0'}$ is finite and non-empty and ends in state $s_n$ and with signature $\twovec{t_n}{t_n'}$, then for every transition $s_n \xrightarrow{m,(1, \twovec{t_n}{t_n'})} s_{n+1} , \bar{t}$ we add the transition $s_0 \xrightarrow{\epsilon} s_{n+1}$ (note that by determinacy of the strategy and restriction to good variable behaviour, this $\epsilon$ transition can be compressed out without loss of determinacy).
\item All transitions on $x$-moves are removed
\item Transitions from final states as required by the inductive hypothesis are added.  This does not affect the language recognised, since the added transitions will require a level-0 data value to be ``in" the relevant copy of $Q_0$, and there can only be one level-0 data value in runs of this automaton.
\end{itemize}

Determinacy of the resulting automaton is inherited from determinism of $\calC_\gamma$ (and thence from $\calA^M_\gamma)$.

\subsection{$\lambda x^\beta . M : \beta \rightarrow \theta$}\label{proof:rmlpstr-lambda}
We have $\seq{\Gamma, x: \beta}{M : \theta}$, and therefore assume there is a family of automata $\calA^M_i$ recognising $\sem{M}$.  
The prearenas for $\sem{\seq{\Gamma, x}{M}}$ and $\sem{\seq{\Gamma}{\lambda x.M}}$ are shown in figure \ref{fig:prearenas-lambda}.  Note that the initial moves in $\sem{\seq{\Gamma, x}{M}}$ contain an $x$-component, so may be considered pairs $({\gamma}, i_x)$, while the initial moves in  $\sem{\seq{\Gamma}{\lambda x.M}}$ contain the same $\Gamma$-component, but no $x$-component.  The move $q_0$ therefore corresponds to the $\Gamma$-component, and the move $q_1$ precisely corresponds to an $x$-move.

$\sem{\seq{\Gamma}{\lambda x.M}}$ is as follows: after an initial move ${\gamma}$, P plays the unique $a_0$-move $\bullet$, and waits for a $q_1$-move.  Once O plays a $q_1$-move $i_x$, P plays as in $\sem{\seq{\Gamma,x}{M}}$ when given an initial move $({\gamma}, i_x)$.
However, as the $q_1$-moves are not initial, it is possible that O will play another $q_1$-move, $i_x'$.  Each time O does this it opens a new thread which P plays as per $\sem{\seq{\Gamma,x}{M}}$ when given initial move $({\gamma}, i_x')$.  Only O may switch between threads, and this can only happen immediately after P plays an $a_i$-move (for any $i$).  
Hence we construct $\calA^{\lambda x.M}_{{\gamma}}$ as follows:
\begin{itemize}
\item The set of states is the disjoint union of the set of non-initial states of each $\calA^M_{(\bar{\gamma},i_x)}$, plus new states $(1)$, $(2)$, and $(3)$.
\item The initial state is $(1)$
\item The final states are those that are final in each $\calA^M_{(\bar{\gamma},i_x)}$, as well as $(1)$ and $(3)$.
\item The transition relation is as follows:
\begin{itemize}
	\item $(1) \xrightarrow{\gamma, (0,\bot)} (2)$
	\item $(2) \xrightarrow{a_0, (0,(2))} (3)$
	\item For each $i_x$, $(3) \xrightarrow{i_x, (1,\twovec{(3)}{\bot})} s_{i_x}$ where $s_{i_x}$ is the secondary state of $\calA^M_{\bar{\gamma},i_x}$
	\item If $s_1 \xrightarrow{m, (j,\bar{s})} s_2, \bar{t}$ is a (non-initial) transition in one of the $\calA^M_{i_x}$, then:
	\begin{itemize}
		\item if $m$ is a $q_i$ or $a_i$ move, $s_1 \xrightarrow{m, (j+1,\twovec{(3)}{\bar{s}})} s_2, \twovec{(3)}{\bar{t}}$ is a transition.
		\item if $m$ is a move in $\sem{\Gamma}$,  $s_1 \xrightarrow{m, (j,{\bar{s}[(3) / s_0]})} s_2, {\bar{t}}[(3) / t_0]$ is a transition.
	\end{itemize}	 
	\item If $s_1$ and $s_1'$ are both (non-initial) final states and $s_1 \xrightarrow{m, (j,\bar{s})} s_2, \bar{t}$ is a transition already given by the above rules, then $s_1' \xrightarrow{m, (j,\bar{s})} s_2, \bar{t}$. (This allows O to switch between threads).
\end{itemize}
\end{itemize}

\subsection{$\letin{x^\beta = M}{N} : \theta $}
Here we have $\seq{\Gamma}{M : \beta}$ and $\seq{\Gamma, x : \beta}{N : \theta}$.  The initial moves of $\sem{\seq{\Gamma, x : \beta}{N : \theta}}$ contain an $x$-component, so we index the family of automata recognising $\sem{\seq{\Gamma, x : \beta}{N : \theta}}$ as $\calA^N_{\gamma,j}$ where $j$ is the $x$-component.  The family of automata recognising $\sem{M}$ are indexed as $\calA^M_\gamma$. 

The strategy $\sem{\letin{x^\beta = M}{N}}$ is essentially a concatenation of the strategies for $\sem{M}$ and $\sem{N}$, with the result of the $\sem{M}$ strategy determining the $x$-component of the initial move of $\sem{N}$.   $\calA^{\letin{x^\beta = M}{N}}_i$ is constructed as follows:

If $\mathcal{L}(\calA^M_i) = \set{\twovec{\gamma}{d} \twovec{j}{d}}$ then $\calA^{\letin{x^\beta = M}{N}}_\gamma = \calA^N_{\gamma,j}$.  Otherwise, by determinacy of the strategy there cannot be a transition from the secondary state to a final state in $\calA^M_\gamma$, and $\calA^{\letin{x^\beta = M}{N}}_\gamma$ is given by:
\begin{itemize}
\item The set states of states is disjoint union of the non-initial states of $\calA^M_i$ and each $\calA^N_{i,j}$, plus new a state $(1)$.
\item The initial state is $(1)$.
\item The final states are those which are final in each $\calA^N_{\gamma,j}$.
\item The transitions are given as follows:
\begin{itemize}
	\item $(1) \xrightarrow{\gamma, (0, \bot)} s_M$ where $s_M$ is the secondary state of $\calA^M_\gamma$
	\item All transitions in $\calA^M_\gamma$ not going to a final state (or from the initial state) are preserved
	\item If $s_1 \xrightarrow{j,(0,s_3)} s_2, t$ is a transition in $\calA^M_\gamma$ with $s_2$ final (in $\calA^M_\gamma$) and $s_{N,j}$ is the secondary state of $\calA^N_{\gamma,j}$:
	\begin{itemize}
		\item if $s_{N,j} \xrightarrow{m,(0,s_{N,j})} s_4, t'$ is in $\calA^N_{\gamma,j}$ then we have the transition $s_1 \xrightarrow{m,(0,s_3)} s_4, t'$.
		\item if $s_{N,j} \xrightarrow{m,(1,\twovec{s_{N,j}}{\bot})} s_4, \bar{t'}$ is in $\calA^N_{\gamma,j}$ then we have the transition $s_1 \xrightarrow{m,(1,\twovec{s_3}{\bot})} s_4, \bar{t'}$.
	\end{itemize}	  
	\item All other transition in each $\calA^N_{\gamma,j}$ are preserved unchanged.
	\item Transitions from final states as required by the inductive hypothesis are added.  This does not affect the language recognised, since the added transitions will require a level-0 data value to be ``in" the relevant copy of $\calA^N_{\gamma,j}$, and there can only be one level-0 data value in runs of this automaton.
\end{itemize}
\end{itemize}
\noindent Determinacy is inherited from $\calA^M_\gamma$ and $\calA^N_{\gamma,j}$.

\subsection{$\letin{x = z y^\beta}{M} : \theta $}
As $x$ must be of type $\beta$ for this to be in $\rmlpstr$, this is essentially the same as the previous case.

\subsection{$\letin{x = z ( \lambda y . M)}{N} : \theta $}
\begin{figure}
	\[\begin{tikzpicture}[arena]		
		\node[clear](i){$(\gamma,i_z)$};
		\node[tri, shape border uses incircle,shape border rotate=45,scale= 0.7](gamma)at (i.south){$\sem{\Gamma}$};
		\node[clear](qz)[below = 3em of i]{$q_z$} edge(i);
		\node[clear](q0z) [below left =of qz] {$q_0'$} edge(qz);
		\node[clear](a0z) [below =of q0z] {$a_0'$} edge(q0z);
		\node[clear, below= 2em of qz](az){$a_z$} edge (qz);
		\node[tri,  shape border rotate=90](theta3) at (a0z.south){$\sem{\theta_1}$};
		\node[clear, right = 6em of qz](a0){$a_0$} edge (i);
		\node[clear, below = of a0](q1){$q_1$} edge (a0);
		\node[clear, below = of q1](a1){$a_1$} edge (q1);
		\node[clear, below = of a1](dots) {$\vdots$} edge (a1);
		\node[clear, below = of dots](qn){$q_n$} edge (dots);
		\node[clear, below = of qn](an){$a_n$} edge (qn);
		\end{tikzpicture}\]
\caption{Prearena for $\sem{\seq{\Gamma, z: (\beta \rightarrow \theta_1) \rightarrow \beta}{\theta}}$}\label{fig:prearena-pstr-abstractiononLHS}
\end{figure}
Here we have $\seq{\Gamma, y: \beta, z: (\beta \rightarrow \theta_1) \rightarrow \beta}{M: \theta_1}$ and $\seq{\Gamma, x: \beta, z: (\beta \rightarrow \theta_1) \rightarrow \beta}{N : \theta}$.  As in the previous cases, plays in $\sem{\letin{x = z ( \lambda y . M)}{N}}$ consist of P playing $\sem{z ( \lambda y . M)}$ until $x$ has been evaluated, and then playing as $N$ with this value of $x$.  The prearena for this case is shown in figure \ref{fig:prearena-pstr-abstractiononLHS}.

Plays in $\sem{\letin{x = z ( \lambda y . M)}{N}}$ start with P playing $q_z$.  O can then either play $q_0'$, starting an $\sem{\lambda y. M}$-thread, or play $a_z$, giving a value for $x$ in the rest of the play.  If O chooses the former, that thread is played as $\sem{\lambda y. M}$, with $q_0'$-moves providing a new value for $y$, until O plays an $a_z$ move.   Once O does play an $a_z$ move, P plays as $\sem{N}$ with the answer O provided as the value for $x$.

In this construction a similar construction to that in \ref{proof:rmlpstr-lambda} will be used, to allow O to interleave plays of $\sem{M}$.  At any point when O would be able to change threads, it is also able to finish evaluating $M$ and give a value for $x$.  Once this happens play continues in the corresponding $\calA^N_{\gamma,j_x}$.
The formal construction for $\calA^{\letin{x = z ( \lambda y . M)}{N}}_{\gamma, i_z}$ is as follows:
\begin{itemize}
\item The set of states consists of:
\begin{itemize}
	\item Fresh states $(1)$, $(2)$, and $(3)$
	\item A copy of the states of each $\calA^M_{\gamma,i_y}$
	\item A copy of the states of each $\calA^N_{\gamma,j_x}$
\end{itemize}
\item The initial state is $(1)$
\item The final states are those which are final in each $\calA^{N}_{\gamma,j_x}$, and $(1)$
\item The transitions are:
\begin{itemize}
	\item $(1) \xrightarrow{\gamma, (0,\bot)} (2)$
	\item $(2) \xrightarrow{q_z, (1,\twovec{(2)}{\bot})} (3)$
	\item For each available $a_z$ move labelled $j_x$ we have the transition $(3) \xrightarrow{j_x, \twovec{(2)}{(3)}} s_{N,j}$ where $s_{N,j}$ is the secondary state of $\calA^N_{\gamma,j_x}$
	\item For each available $q_0'$ move labelled $i_y$, we have $(3) \xrightarrow{i_y, (2,\threevec{(2)}{(3)}{\bot})} s_M$ where $s_M$ is the secondary state of $\calA^{M}_{\gamma,i_y}$
	\item If $s_1 \xrightarrow{m, (j,\bar{s})} s_2, \bar{t}$ is a (non-initial) transition in one of the $\calA^M_{i_y}$, then:
	\begin{itemize}
		\item if $m$ is a $q_i$ or $a_i$ move in $\sem{\theta1}$, $s_1 \xrightarrow{m, (j+2,\threevec{(2)}{(3)}{\bar{s}})} s_2, \threevec{(2)}{(3)}{\bar{t}}$ is a transition.
		\item if $m$ is a move in $\sem{\Gamma}$,  $s_1 \xrightarrow{m, (j,{\bar{s}[(2) / s_0]})} s_2, {\bar{t}}[(2) / t_0]$ is a transition.
	\end{itemize}
	\item If $q_1 \xrightarrow{m,(k, \bar{s}) q_2, \bar{t}}$ is a transition already defined by one of these, and $q_1$ is either $(3)$ or a (non-initial) final state in one of the $\calA^M_{\gamma, i_y}$, and $q_3$ is a (non-initial) final state in one of the $\calA^M_{\gamma, i_y}$ then we have the transition $q_3 \xrightarrow{m,(k, \bar{s}) q_2, \bar{t}}$.
	
	\item For each transition $s_{N,j} \xrightarrow{m, (k, \bar{s})} q, \bar{t}$ in each $\calA^N_{\gamma,j_x}$, where $s_{N,j}$ is the secondary state of $\calA^N_{\gamma,j_x}$, we have the transition $s_{N,j} \xrightarrow{m, (k, \bar{s}[(2) / s_0])} q, \bar{t}$
	\item All other transitions in each $\calA^N_{\gamma,j_x}$ are left unchanged
	\item Transitions from final states as required by the inductive hypothesis are added.  This does not affect the language recognised, since the added transitions will require a level-0 data value to be ``in" the relevant copy of $\calA^N_{\gamma,j_x}$, and there can only be one level-0 data value in runs of this automaton.
\end{itemize}
\end{itemize}

\noindent Determinism is inherited from the constituent automata.	

\subsection{$\letin{x = z \makevar{\lambda u^\unitt . M_1}{\lambda v^\intt . M_2}}{N} : \theta $}
This is very similar to the previous case: the difference is that the $q_0'$ moves from the last case can now be either $read$ or $write(j)$, leading to playing as either $\sem{M_1}$ or $\sem{M_2}$ respectively.  The formal construction is almost identical to that given above.

\section{Proof of Theorem \ref{thm:forml-proof}}\label{appendix:rforml}

Given a $\rforml$ term-in-context $\seq{\Gamma}{M}$ we construct a Deterministic Weak NDCMA $\calA_{\seq{\Gamma}{M}}$ recognising, as a language, $\comp{\sem{\seq{\Gamma}{M}}}$.  By the full abstraction theorem, observational equivalence can then be checked by testing the corresponding automata for equivalence.

The shape of the pre-arena for terms $\sem{\seq{\Gamma}{M}}$ in $\rforml$ is shown in figure \ref{fig:1frorml-prearenas}.  The moves in section $A$ of the prearena correspond to $M$, while moves in sections $B$ and $C$ correspond to $\Gamma$.

A play $p$ in $\sem{\seq{\Gamma}{\alpha(n)}}$ is represented in the data language as a word $w$ where the string projection of $w$ is equal to the underlying sequence of moves in $p$.  Pointers are only ambiguous for question moves (as for answers well-bracketing is enough to ensure justification is clear).  Pointers for questions are represented in the following manner:  
\begin{itemize}
\item Initial questions (of which there is precisely one, at the beginning of the play) take a fresh level-$0$ data value.  
\item If $a$ is an answer-move in the play, then the corresponding letter in the word will be $\twovec{a}{d}$ where $d$ is the same data value as the answer's justifier.
\item Question moves in section A of the arena above take a fresh data value, $d$, such that $pred(d) = d'$ where $d'$ is the data value of the justifier.  These data values will be enough to determine the justifiers.
\item All other moves (i.e. those in sections B and C) take the data value of the most recent move in A (or the initial move, if no move in A has yet been made).  Moves in B will have their pointers represented using the ``tagging" of source- and target-moves, as used in \cite{HMO11} for $\rmlostr$.  We will not encode pointers of such moves justified by the initial move (i.e. $q^{(1)}$ moves), as they are unambiguously justified. 
\end{itemize}

\textbf{Reduction from $\rforml$}
The reduction is inductive on the construction of the canonical form.  We make the construction indexed by initial moves, with each automaton $\calA_i$ recognising the appropriate language restricted to the initial move $i$.  The construction to combine these into one automaton as per the specification above is a  straightforward union of the automata and merging of the initial states.

Our inductive hypothesis is slightly stronger than that the constructed automata recognises the appropriate languages.  We also require the following conditions on the automaton $\calA^M_i$:
\begin{itemize}
\item Initial states are never revisited (or have data values assigned to them)
\item The automaton is deterministic
\item Each state can only ever ``hold" data values of one, fixed, level.
\item There is precisely one transition from the initial state, labelled $i, (0,\bot)$.  We will call the target state of this transition the ``secondary state" of the automaton.  Further, this is the only transition in the automaton with signature $(0,\bot)$.
\item If $q$ and $q'$ are (non-initial) final states in the automaton, then if there is a transition $(q,a,\xi,p,\xi')$ then $(q',a,\xi,p,\xi')$ is also a transition.
\end{itemize} 

%
%
%

For the cases $(): \unitt$, $i : \intt$, $x^\beta: \beta$, $\suc{x^\intt} : \intt$, and $\pre{x^\intt} : \intt$, the constructions are exactly as in $\rmlpstr$.  We deal with the remaining cases here:

\subsection{$x^{\intreft} := y^\intt : \unitt$}
Here we have $\seq{\Gamma}{x^{\intreft} := y^\intt}$, so $x : \intreft$ and $y: \intt$ are in $\Gamma$.  Thus the initial moves have a $y$-component, say $j$.  Thus the language recognised by $\calA_{(\bar{\gamma},j)}$ is just $\set{ \twovec{(\bar{\gamma},j)}{d} \twovec{write_x (j)}{d} \twovec{ok_x}{d} \twovec{\bullet}{d} \: | \: d, \in \dataset \text{ and } d \text{ is level-0}}$.  This is recognised by the following automaton:

\begin{tikzpicture}[->,>=stealth',shorten >=1pt,auto,node distance=3cm,
                    semithick,]
                   
  \node[state,initial,accepting] (1) {$s_1$};  
  \node[state] (2) [right of=1] {$s_2$};
  \node[state] (3) [right of=2] {$s_3$};
  \node[state] (4) [right of=3] {$s_4$};
  \node[accepting,state] (5) [right of=4] {$s_5$};
  \path (1) edge              node {\small $(\bar{\gamma},j) , (0,\bot)$} 		(2)
   (2) edge              node {\small $w_x(j), (0, s_2)$} 		(3)
   (3) edge              node {\small $ok_x, (0,s_3)$} 		(4)
   (4) edge              node {\small $\bullet, (0,s_4)$} 		(5);
\end{tikzpicture}

\subsection{$\deref{x^\intreft} : \intt$}
This is similar to the previous case, only the value to return is given by O's play in the $x$-section of $\Gamma$.  The language recognised by $\calA_\gamma$ is just: $\set{ \twovec{\gamma}{d} \twovec{read_x}{d} \twovec{j_x}{d} \twovec{j}{d} \: | \: d \in \dataset \text{ and } d \text{ is level-0}}$.  The automaton is thus similar to that given above, except that from state $s_3$ the automaton splits into different states for each possible answer $j_x$.

\subsection{$\cond{x^\beta}{M}{N} : \theta$}
The initial move contains an $x$-component.  If this $x$-component is $0$ then the automaton is as the as the automaton for $N$, otherwise it is as the automaton for $M$.

\subsection{$\makevar{\lambda x^\unitt . M}{\lambda y^\intt . N} : \intreft$}\label{proof:1rforml-mkvar}
Here we have $\seq{\Gamma, x : \unitt}{M : \intt}$ and $\seq{\Gamma, y : \intt}{N: \unitt}$, and this ``bad-variable" construction uses these methods as read- and write-methods respectively.  The string projection of the language for $\sem{\makevar{\lambda x^\unitt . M}{\lambda y^\intt . N}}$ is then
$$ \gamma \cdot \bullet \cdot ( read \cdot L_M + \sum_j write(j) \cdot L^j_N )^*$$
Where $L_M$ is the language for $\sem{M}$ without the initial move, and $L^j_N$ is the language $\sem{N}$ when $y = j$, without the initial move. Note that the representing automata are level-0.

For an initial move $\bar{\gamma}$, we make the following construction of $\calA^{\makevar{\lambda x. M}{\lambda y . N}}_{\bar{\gamma}}$:

\begin{itemize}
\item The set of states is the disjoint union of the states of $\calA^{M}_{\bar{\gamma}}$, and each $\calA^N_{(\bar{\gamma},j)}$, minus the initial states, plus additional states $(1)$, $(2)$, and $(3)$.
\item The initial state is the state $(1)$.
\item The final states are those which are final in the constituent automata $\calA^{M}_{\bar{\gamma}}$ and each $\calA^N_{(\bar{\gamma},j)}$, and $(1)$ and $(3)$.
\item The transition relation is given as follows:
\begin{itemize}
	\item $(1) \xrightarrow{\bar{\gamma}, (0, \bot)} (2)$
	\item $(2) \xrightarrow{\bullet, (0, (2))} (3)$
	\item $(3) \xrightarrow{read, (1, \twovec{(3)}{\bot})} s_M$ where $s_M$ is the secondary state of $\calA^M_{\bar{\gamma}}$
	\item $(3) \xrightarrow{write(j), (1, \twovec{(3)}{\bot})} s_{N,j}$ where $s_{N,j}$ is the secondary state of $\calA^N_{(\bar{\gamma},j)}$
	\item For all transitions $s_1 \xrightarrow{m, (0, \xi)} s_2$ in a constituent automaton (not including initial transition), we have the transition $s_1 \xrightarrow{m, (1, \twovec{(3)}{\xi})} s_2$.
	\item From each final state, $s$, in one of the constituent automata, we add transitions $s \xrightarrow{read, (1, \twovec{(3)}{\bot})} s_M$ and $s \xrightarrow{write(j), (1, \twovec{(3)}{\bot})} s_{N,j}$ (where $s_M$ and $s_{N,j}$ are as before)
\end{itemize}
\end{itemize}

\noindent We note that determinism is inherited from the constituent automata.  Further, the only transitions from final states we need to add for the inductive hypothesis have already been added.

\subsection{$\while{M}{N} : \unitt$}\label{proof:1rforml-while}
The strategy $\sem{\while{M}{N}}$ plays as if playing $M$ until the final move would be made.  If this would be 0, P gives the $\bullet$ answer to the initial move, and stops.  Otherwise it plays as if playing $N$, until the final move would be made, when it starts as if playing $M$ again.  Note that $\calA^M_\gamma$ and $\calA^N_\gamma$ are both level-0. The automata $\calA^{\while{M}{N}}_\gamma$ is thus given by:
\begin{itemize}
\item The set of states is given by the disjoint union of the set of states of $\calA^M_\gamma$ and $\calA^N_\gamma$, without the initial states, plus new states $(1)$ and $(2)$.
\item The initial state is $(1)$.
\item The final states are $(1)$ and $(2)$.
\item The transitions are given as follows:
\begin{itemize}
	\item $(1) \xrightarrow{\bar{\gamma}, (0, \bot)} s_M$ where $s_M$ is the secondary state of $\calA^M_\gamma$.
	\item if $s'$ is a final state of $\calA^M_\gamma$ and $s \xrightarrow{m, (0,\xi)} s'$ is a transition in $\calA^M_\gamma$, with $m \neq 0$, we have the transition $s \xrightarrow{\epsilon} s_N$ where $s_N$ is the secondary state of $\calA^N_\gamma$. (We can compress the silent transition $\epsilon$ out, since by determinism of the strategy this is the only transition from $s$ in $\calA^M_\gamma$.)
	\item if $s'$ is a final state of $\calA^M_\gamma$ and $s \xrightarrow{m, (0,\xi)} s'$ is a transition in $\calA^M_\gamma$, with $m = 0$, we have the transition $s \xrightarrow{\bullet, (0,\xi)} (2)$
	\item if $s'$ is a final state of $\calA^N_\gamma$ and $s \xrightarrow{m, (0,\xi)} s'$ is a transition in $\calA^N_\gamma$, we have the transition $s \xrightarrow{\epsilon} s_M$ where $s_M$ is the secondary state of $\calA^M_\gamma$. (We can compress the silent transition $\epsilon$ out, since by determinism of the strategy this is the only transition from $s$ in $\calA^N_\gamma$.)
\end{itemize}
\end{itemize}

\noindent Determinacy is inherited from the constituent automata, and there are no transitions from final states that need be added. 

\subsection{$\letin{x = \newc}{M} : \theta$}
This is similar to the construction for $\rmlpstr$ in appendix \ref{appendix:pstrict}, but this time the value of the variable will be stored just by the level-0 data value.  This will correctly capture the scope of the variable.

We assume we have a family of automata, $\calA^M_i$, recognising the strategy $\sem{\seq{\Gamma, x: \intreft}{M: \theta}}$. 
$\sem{\seq{\Gamma}{\letin{x = \newc}{M} : \theta}}$ is constructed by restricting behaviour of $x$ to ``good variable" behaviour (i.e. after a read-move the response is an immediate reply of the last integer written to the variable), and then hiding those moves.  The automata construction is done in these two stages.

\textbf{Restriction to good-variable behaviour.} \quad 
Assume the finitary fragment we are using is $\set{0,1, \dots, k}$.  By our inductive hypothesis, we know that each state can only 'hold' data values of one level: let $Q_0$ be the set of states of $\calA^M_\gamma$ which hold level-0 data values, let $Q_{\geqslant 1}$ be the set of states of $\calA^M_\gamma$ which hold data values of level $\geqslant 1$, so the states of $\calA^M_\gamma$ are partitioned into $Q_0$, $Q_{\geqslant 1}$, and the initial state $q_I$.  We construct $\calC_\gamma$ as follows:
\begin{itemize}
\item The states of the automaton are $\set{q_I} \uplus Q_{\geqslant 1} \uplus (Q_0 \times \set{0,1, \dots, k})$
\item The final states are those which are final in $\calA^M_\gamma$, and those which are final in $\calA^M_\gamma$ paired with any integer.
\item The initial state is $q_I$, the initial state in $\calA^M_\gamma$.
\item The transitions are given as follows:
\begin{itemize}
	\item $q_I \xrightarrow{q_0, (0,\bot)} (s_M,0)$ where $s_M$ is the secondary state of $\calA^M_\gamma$
	\item If $s_1 \xrightarrow{m, (0, s_3)} s_2, t$ is in $\calA^M_\gamma$, where $m$ is not an $x$-$write$ move or a response to an $x$-$read$ move, then $s_1, s_2 \in Q_0$, and we have $(s_1, i) \xrightarrow{m, (0, (s_3, i))} (s_2, i), (t,i)$ for each $i$
	\item If $s_1 \xrightarrow{m, (k, \twovec{s_3}{\bar{\xi}})} s_2, \twovec{t}{\bar{t}}$ is in $\calA^M_\gamma$ (where $k \geq 1$), where $m$ is not an $x$-$write$ move or a response to an $x$-$read$ move, then $s_1, s_2 \in Q_{\geqslant 1}$, and we have $s_1 \xrightarrow{m, (k, \twovec{(s_3, i)}{\bar{\xi}})} s_2, \twovec{(t,i)}{\bar{t}}$ for each $i$
	
	\item For each $j$, if $s_1 \xrightarrow{write_x(j), (0, s_3)} s_2, t$ is in $\calA^M_\gamma$, then	$s_1, s_2 \in Q_0$, and we have $(s_1, i) \xrightarrow{write_x(j), (0, (s_3, i))} (s_2, j), (t,j)$ for each $i$
	\item For each $j$, if $s_1 \xrightarrow{write_x(j), (k, \twovec{s_3}{\bar{\xi}})} s_2, \twovec{t}{\bar{t}}$ is in $\calA^M_\gamma$ (where $k \geq 1$), then $s_1, s_2 \in Q_{\geqslant 1}$, and we have $s_1 \xrightarrow{write_x(j), (k, \twovec{(s_3, i)}{\bar{\xi}})} s_2, \twovec{(t,j)}{\bar{t}}$ for each $i$
	
	\item For each response to an $x$-read move, $j_x$, if $s_1 \xrightarrow{j_x, (0, s_3)} s_2, t$ is in $\calA^M_\gamma$, then $s_1, s_2 \in Q_0$, and we have $(s_1, j) \xrightarrow{j_x, (0, (s_3, j))} (s_2, j), (t,j)$
	\item For each response to an $x$-read move, $j_x$, if $s_1 \xrightarrow{j_x, (k, \twovec{s_3}{\bar{\xi}})} s_2, \twovec{t}{\bar{t}}$ is in $\calA^M_\gamma$ (where $k \geq 1$), then $s_1,s_2 \in Q_{\geqslant 1}$, and we have $s_1 \xrightarrow{j_x, (k, \twovec{(s_3, j)}{\bar{\xi}})} s_2, \twovec{(t,j)}{\bar{t}}$
\end{itemize}
\end{itemize}

\textbf{Hiding} \quad 
$\calA^{\letin{x = \newc}{M} }_\gamma$ is constructed from $\calC_\gamma$ as follows:

If we are in a configuration $(s_1, f)$ of $\calC_i$ where we can perform a transition $s_1 \xrightarrow{m_x, (j,\bar{s})} s_2, \bar{t}$ where $m_x$ is an $x$-move then by determinacy of strategies combined with the restriction to good variable behaviour, it is the only possible transition from this configuration.  Further, we note that using only $x$-transitions cannot lead to a change in data-value being read. Thus for every state $s_0$ of $\calC_\gamma$ and every possible ``signature" $\bar{\xi}_0$, there is a unique maximal (and not necessarily finite) sequence of transitions:
$$ s_0 \xrightarrow{m_0, (k, \bar{\xi}_0)} s_1, \bar{\xi}_1 \xrightarrow{m_1, (k, \bar{\xi}_1)} s_2, \bar{\xi}_2 \xrightarrow{m_2, (k, \bar{\xi}_2)} \dots $$
where each $m_i$ is an $x$-move.

From each $\calC_\gamma$ we construct the automaton $\calA^{\letin{x = \newc}{M}}_\gamma$ by considering where this sequence terminates for each state.  Everything is the same as in $\calC_\gamma$ except for the transition relation, which is altered as follows:
\begin{itemize}
\item If the maximal sequence of $x$-moves with signature $\bar{\xi}_0$ out of state $s_0$ is empty then all transitions requiring signature $\bar{\xi}_0$ out of $s_0$ are unchanged.
\item If the maximal sequence out of $s_0$ with signature $\bar{\xi}_0$ is finite and non-empty and ends in state $s_n$ and with signature $\bar{\xi}_n$, then for every transition $s_n \xrightarrow{m,(k, \bar{\xi}_n)} s_{n+1} , \bar{t}$ we add the transition $s_0 \xrightarrow{m,(k, \bar{\xi_0})} s_{n+1} , \bar{t}$.
\item All transitions on $x$-moves are removed
\item Transitions from final states as required by the IH are added.  This does not affect the language recognised, since the added transitions will require a level-0 data value to be ``in" the relevant location, and there can only be one level-0 data value in runs of this automaton.
\end{itemize}

Determinacy of the resulting automaton is inherited from determinism of $\calC_\gamma$ (and thence from $\calA^M_\gamma)$.    

\subsection{$\lambda x^\beta . M : \theta$}
We have $\seq{\Gamma, x: \beta}{M : \theta'}$, and therefore assume there is a family of automata $\calA^M_i$ recognising $\sem{M}$.  
The prearenas for $\sem{\seq{\Gamma, x}{M}}$ and $\sem{\seq{\Gamma}{\lambda x.M}}$ are shown in figure \ref{fig:prearenas-lambda}.  Note that the initial moves in $\sem{\seq{\Gamma, x}{M}}$ contain an $x$-component, so may be considered pairs $({\gamma}, i_x)$, while the initial moves in  $\sem{\seq{\Gamma}{\lambda x.M}}$ contain the same $\Gamma$-component, but no $x$-component.  The move $q_0$ therefore corresponds to the $\Gamma$-component, and the move $q_1$ precisely corresponds to an $x$-move.

$\sem{\seq{\Gamma}{\lambda x.M}}$ is as follows: after an initial move ${\gamma}$, P plays the unique $a_0$-move $\bullet$, and waits for a $q_1$-move.  Once O plays a $q_1$-move $i_x$, P plays as in $\sem{\seq{\Gamma,x}{M}}$ when given an initial move $({\gamma}, i_x)$.
However, as the $q_1$-moves are not initial, it is possible that O will play another $q_1$-move, $i_x'$.  Each time O does this it opens a new thread which P plays as per $\sem{\seq{\Gamma,x}{M}}$ when given initial move $({\gamma}, i_x')$.  Only O may switch between threads, and this can only happen immediately after P plays an $a_i$-move (for any $i$).  Thus we construct $\calA^{\lambda x.M}_{{\gamma}}$ as follows:
\begin{itemize}
\item The set of states is the disjoint union of the set of non-initial states of each $\calA^M_{(\bar{\gamma},i_x)}$, plus new states $(1)$, $(2)$, and $(3)$.
\item The initial state is $(1)$
\item The final states are those that are final in each $\calA^M_{(\bar{\gamma},i_x)}$, as well as $(1)$ and $(3)$.
\item The transition relation is as follows:
\begin{itemize}
	\item $(1) \xrightarrow{\gamma, (0,\bot)} (2)$
	\item $(2) \xrightarrow{a_0, (0,(2))} (3)$
	\item For each $i_x$, $(3) \xrightarrow{i_x, (1,\twovec{(3)}{\bot})} s_{i_x}$ where $s_{i_x}$ is the secondary state of $\calA^M_{\bar{\gamma},i_x}$
	\item If $s_1 \xrightarrow{m, (j,\bar{s})} s_2, \bar{t}$ is a (non-initial) transition in one of the $\calA^M_{i_x}$, then $s_1 \xrightarrow{m, (j+1,\twovec{(3)}{\bar{s}})} s_2, \twovec{(3)}{\bar{t}}$ is a transition.
	
\item If $s_1$ and $s_1'$ are both (non-initial) final states and $s_1 \xrightarrow{m, (j,\bar{s})} s_2, \bar{t}$ is a transition already given by the above rules, then $s_1' \xrightarrow{m, (j,\bar{s})} s_2, \bar{t}$.
\end{itemize}
\end{itemize}

\subsection{$\letin{x^\beta = M}{N} : \theta $}
This is very similar to the equivalent case in appendix \ref{appendix:pstrict}.

The strategy $\sem{\letin{x^\beta = M}{N}}$ is a concatenation of the strategies for $\sem{M}$ and $\sem{N}$, with the result of the $\sem{M}$ strategy determining the $x$-component of the initial move of $\sem{N}$.
We have $\seq{\Gamma}{M : \beta}$ and $\seq{\Gamma, x : \beta}{N : \theta}$.  The initial moves of $\sem{\seq{\Gamma, x : \beta}{N : \theta}}$ contain an $x$-component, so we index the family of automata recognising $\sem{\seq{\Gamma, x : \beta}{N : \theta}}$ as $\calA^N_{\gamma,j}$ where $j$ is the $x$-component.  The family of automata recognising $\sem{M}$ are indexed as $\calA^M_\gamma$. 
$\calA^{\letin{x^\beta = M}{N}}_i$ is constructed as follows:

If $\mathcal{L}(\calA^M_i) = \set{\twovec{\gamma}{d} \twovec{j}{d}}$ then $\calA^{\letin{x^\beta = M}{N}}_\gamma = \calA^N_{\gamma,j}$.  Otherwise, by determinacy of the strategy there cannot be a transition from the secondary state to a final state in $\calA^M_\gamma$, and $\calA^{\letin{x^\beta = M}{N}}_\gamma$ is given by:
\begin{itemize}
\item The set states of states is disjoint union of the non-initial states of $\calA^M_i$ and each $\calA^N_{i,j}$, plus new a state $(1)$.
\item The initial state is $(1)$.
\item The final states are those which are final in each $\calA^N_{\gamma,j}$.
\item The transitions are given as follows:
\begin{itemize}
	\item $(1) \xrightarrow{\gamma, (0, \bot)} s_M$ where $s_M$ is the secondary state of $\calA^M_\gamma$
	\item All transitions in $\calA^M_\gamma$ not going to a final state (or from the initial state) are preserved
	\item If $s_1 \xrightarrow{j,(0,s_3)} s_2, t$ is a transition in $\calA^M_\gamma$ with $s_2$ final (in $\calA^M_\gamma$) and $s_{N,j}$ is the secondary state of $\calA^N_{\gamma,j}$,  and $s_{N,j} \xrightarrow{m,(0,s_{N,j})} s_4, t'$ is in $\calA^N_{\gamma,j}$ then we have the transition $s_1 \xrightarrow{m,(0,s_3)} s_4, t'$.
	\item All other transition in each $\calA^N_{\gamma,j}$ are preserved unchanged.
	\item Transitions from final states as required by the inductive hypothesis are added.  This does not affect the language recognised, since the added transitions will require a level-0 data value to be ``in" the relevant copy of $\calA^N_{\gamma,j}$, and there can only be one level-0 data value in runs of this automaton.
\end{itemize}
\end{itemize}
\noindent Determinacy is inherited from $\calA^M_\gamma$ and $\calA^N_{\gamma,j}$.

\subsection{$\letin{x = z y^\beta}{M} : \theta $}
We assume $x$ is not of type $\beta$, as otherwise this could be handled by the previous construction.

We have $\seq{\Gamma, x : \theta', z: \beta \rightarrow \theta', y : \beta}{M: \theta}$.  Plays in $\sem{\letin{x = z y^\beta}{M}}$ begin with P copying the $y$-component of the initial move into the $z$-component, and O must respond with the unique answer, $\bullet_z$ (which corresponds to the initial move of $\sem{\theta'}$).  Play then continues as $\sem{M}$ except that all $x$-moves are relabelled as $z$-moves, hereditarily justified by the occurrence of $\bullet_z$ O was forced to play.  The pointers for moves justified by $\bullet_z$ will have to be made explicit as part of the construction.

$\calA^{\letin{x = z y^\beta}{M}}_{\gamma, i_y, i_z}$ is then constructed as follows:
\begin{itemize}
\item The states are two copies of the non-initial states of $\calA^M_{\gamma, i_y, i_z, \bullet_x}$ (where $\bullet_x$ is the move $\bullet_z$ that O will be forced to play, relabelled as an $x$-move) plus new states $(1)$, $(2)$ and $(3)$.  The second copy of $\calA^M_{\gamma, i_y, i_z, \bullet_x}$ will be used to encode P-pointers, so we write state $s$ in the second copy as $\ptarget{s}$.
\item The initial state is $(1)$.
\item The final states are those final in $\calA^M_{\gamma, i_y, i_z, \bullet_x}$, and $(1)$.
\item The transitions are as follows:
\begin{itemize}
	\item $(1) \xrightarrow{(\gamma,i_y,i_z), (0, \bot)} (2)$
	\item $(2) \xrightarrow{j_z, (0, (2))} (3)$ where $j_z$ is the initial move for $y$ copied into the $z$-component.
	\item $(3) \xrightarrow{\bullet_z, (0, (3))} s_M$ and $(3) \xrightarrow{\ptarget{\bullet_z}, (0, (3))} \ptarget{s_M}$ where $s_M$ is the secondary state of $\calA^M_{\gamma, i_y, i_z, \bullet_x}$.
	\item $s_1 \xrightarrow{m, (k, \bar{s})} s_2, \bar{t}$ is a transition in $\calA^M_{\gamma, i_y, i_z, \bullet_x}$ and $m$ is not an $x$-move, then we have the transitions $s_1 \xrightarrow{m, (k, \bar{s})} s_2, \bar{t}$ and $\ptarget{s_1} \xrightarrow{m, (k, \ptarget{\bar{s}})} \ptarget{s_2}, \ptarget{\bar{t}}$  (where $\ptarget{\bar{s}}$ replaces each element of $s$ of $\bar{s}$ with $\ptarget{s}$).
	\item $s_1 \xrightarrow{m_x, (k, \bar{s})} s_2, \bar{t}$ is a transition in $\calA^M_{\gamma, i_y, i_z, \bullet_x}$ and $m_x$ is a non-initial $x$-move, then we have the transitions $s_1 \xrightarrow{m_z, (k, \bar{s})} s_2, \bar{t}$ and $\ptarget{s_1} \xrightarrow{m_z, (k, \ptarget{\bar{s}})} \ptarget{s_2}, \ptarget{\bar{t}}$, where $m_z$ is the relabelling of $m_x$ into the $z$-component.
	\item $s_1 \xrightarrow{m_x, (k, \bar{s})} s_2, \bar{t}$ is a transition in $\calA^M_{\gamma, i_y, i_z, \bullet_x}$ and $m$ is the initial $x$-move, then we have the transitions $s_1 \xrightarrow{m_z, (k, \bar{s})} s_2, \bar{t}$ and $\ptarget{s_1} \xrightarrow{m_z, (k, \ptarget{\bar{s}})} \ptarget{s_2}, \ptarget{\bar{t}}$ and $\ptarget{s_1} \xrightarrow{\psource{m_z}, (k, \ptarget{\bar{s}})} \ptarget{s_2}, \ptarget{\bar{t}}$, where $m_z$ is the relabelling of $m_x$ into the $z$-component.
	\item Transitions from final states as required by the inductive hypothesis are added.  This does not affect the language recognised, since the added transitions will require a level-0 data value to be ``in" the relevant location, and there can only be one level-0 data value in runs of this automaton.
\end{itemize}
\end{itemize}

\noindent Determinism is inherited from the constituent automaton.

\subsection{$\letin{x = z ( \lambda y . M)}{N} : \theta $} 
Here we have $\seq{\Gamma, y: \beta, z: (\beta \rightarrow \beta) \rightarrow \theta_1}{M: \beta}$ and $\seq{\Gamma, x: \theta_1, z: (\beta \rightarrow \beta) \rightarrow \theta_1}{N : \theta}$.  The prearena is as follows:

	
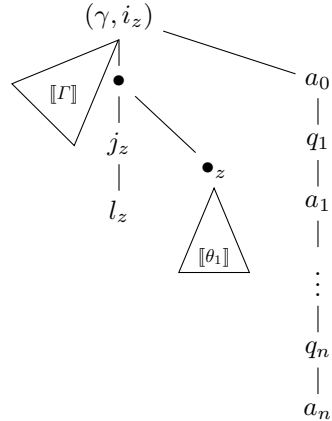
\[\begin{tikzpicture}[arena]		
		\node[clear](i){$(\gamma,i_z)$};
		\node[tri, shape border uses incircle,shape border rotate=45,scale= 0.7](gamma)at (i.south){$\sem{\Gamma}$};
		\node[clear](bullet)[below =of i]{$\bullet$} edge(i);
		\node[clear](jz) [below =of bullet] {$j_z$} edge(bullet);
		\node[clear](kz) [below =of jz] {$l_z$} edge(jz);
		\node[clear, below right = 3em of bullet](j){$\bullet_z$} edge (bullet);
		\node[tri,  shape border rotate=90,scale= 0.7](theta3) at (j.south){$\sem{\theta_1}$};
		\node[clear, right = 6em of bullet](a0){$a_0$} edge (i);
		\node[clear, below = of a0](q1){$q_1$} edge (a0);
		\node[clear, below = of q1](a1){$a_1$} edge (q1);
		\node[clear, below = of a1](dots) {$\vdots$} edge (a1);
		\node[clear, below = of dots](qn){$q_n$} edge (dots);
		\node[clear, below = of qn](an){$a_n$} edge (qn);
		\end{tikzpicture}\]
		
Plays in $\sem{\letin{x = z ( \lambda y . M)}{N}}$ start with P playing $\bullet$.  O can then either play $j_z$, starting an $\sem{M}$-thread, or play $\bullet_z$, the initial $x$-move.  If O chooses the former, that thread is played to completion, as in $\sem{M}$.  Once this is finished, (with P playing $k_z$ as the final move), we return to the situation where O can play either $j_z$ or $\bullet_z$.  Once O does play $\bullet_z$, P plays as $\sem{N}$, except that all $x$-moves are renamed to $z$-moves (justified by $\bullet_z$).  Further, whenever P plays in $x$ (which becomes a $z$-move), O can again play $j_z$ and start an $\sem{M}$ thread.

The automaton $\calA^{\letin{x = z ( \lambda y . M)}{N}}_{\gamma, i_z}$ is constructed as follows:
\begin{itemize}
\item The set of states consists of:
\begin{itemize}
	\item Fresh states $(1)$, $(2)$, and $(3)$
	\item Two copies of the set of non-initial states of $\calA^{N}_{\gamma,\bullet_x, i_z}$, the second marked as $\ptarget{s}$
	\item define $\mathcal{S}$, the set of states from which an $\sem{M}$-thread can be opened, as
	$$ \mathcal{S} = \set{(3)} \uplus \set{ r \: : \: ( r = s \text{ or } r = \ptarget{s} ) \text{ and } t \xrightarrow{m_x} s \text{ in }  \calA^{N}_{\gamma,\bullet_x, i_z} \text{ with } m_x \text{ a P-} x \text{ move}}$$
	We then take states $(s,t)$ where $s$ is a state in some $\calA^M_{\gamma, i_y, i_z}$ and $t \in \mathcal{S}$
\end{itemize}
\item The initial state is $(1)$
\item The final states are those which are final in $\calA^{N}_{\gamma,\bullet_x, i_z}$ (both tagged and untagged), and $(1)$
\item The transitions are:
\begin{itemize}
	\item $(1) \xrightarrow{(\gamma, i_x), (0,\bot)} (2)$
	\item $(2) \xrightarrow{\bullet, (0,(2))} (3)$
	\item $(3) \xrightarrow{\bullet_z, (0,(3))} s_N$ and $(3) \xrightarrow{\ptarget{\bullet_z}, (0,(3))} \ptarget{s_N}$, where $s_N$ is the secondary state of $\calA^{N}_{\gamma,\bullet_x, i_z}$
	
	\item If $s_1 \xrightarrow{m, (k, \bar{s})} s_2, \bar{t}$ is a transition in $\calA^{N}_{\gamma,\bullet_x, i_z}$ and $m$ is not an $x$-move, then we have the transitions $s_1 \xrightarrow{m, (k, \bar{s})} s_2, \bar{t}$ and $\ptarget{s_1} \xrightarrow{m, (k, \ptarget{\bar{s}})} \ptarget{s_2}, \ptarget{\bar{t}}$ 
	\item If $s_1 \xrightarrow{m_x, (k, \bar{s})} s_2, \bar{t}$ is a transition in $\calA^{N}_{\gamma,\bullet_x, i_z}$ and $m_x$ is a non-initial $x$-move, then we have the transitions $s_1 \xrightarrow{m_z, (k, \bar{s})} s_2, \bar{t}$ and $\ptarget{s_1} \xrightarrow{m_z, (k, \ptarget{\bar{s}})} \ptarget{s_2}, \ptarget{\bar{t}}$, where $m_z$ is the relabelling of $m_x$ into the $z$-component.
	\item If $s_1 \xrightarrow{m_x, (k, \bar{s})} s_2, \bar{t}$ is a transition in $\calA^{N}_{\gamma,\bullet_x, i_z}$ and $m$ is the initial $x$-move, then we have the transitions $s_1 \xrightarrow{m_z, (k, \bar{s})} s_2, \bar{t}$ and $\ptarget{s_1} \xrightarrow{m_z, (k, \ptarget{\bar{s}})} \ptarget{s_2}, \ptarget{\bar{t}}$ and $\ptarget{s_1} \xrightarrow{\psource{m_z}, (k, \ptarget{\bar{s}})} \ptarget{s_2}, \ptarget{\bar{t}}$, where $m_z$ is the relabelling of $m_x$ into the $z$-component.
	
	\item If $s \in \mathcal{S}$ then for all transitions $t \xrightarrow{m_x, (k, \bar{t})} s,\bar{s}$ we have the transition $s \xrightarrow{j_z, (k, \bar{s})} (q_{M,j},s), \bar{s}_{q_{M,j}}$, where $q_{M,j}$ is the secondary state of $\calA^M_{\gamma, j_y, i_z}$, and $\bar{s}_{q_{M,j}}$ is the same as $\bar{s}$ but with the last element paired with $q_{M,j}$.  Further:
	\begin{itemize}
		\item If $p_1 \xrightarrow{m, (0,p_1')} p_2, p_2'$ is in $\calA^M_{\gamma, j_y, i_z}$ where $p_2$ is not final (in $\calA^M_{\gamma, j_y, i_z}$), we have $(p_1, s) \xrightarrow{m, (k, \bar{s}_{p_1'})} (p_1,s), \bar{s}_{p_2'}$ 
		\item If $p_1 \xrightarrow{l, (0,p_1')} p_2, p_2'$ is in $\calA^M_{\gamma, j_y, i_z}$ where $p_2$ is final (in $\calA^M_{\gamma, j_y, i_z}$), we have $(p_1, s) \xrightarrow{m, (k, \bar{s}_{p_1'})} s, \bar{s}$ 
	\end{itemize}
	\item Transitions from final states as required by the IH are added.  This does not affect the language recognised, since the added transitions will require a level-0 data value to be ``in" the sub-automaton, and there can only be one level-0 data value in runs of this automaton.
\end{itemize}
\end{itemize}

\noindent Determinism is inherited from the constituent automata.	

\subsection{$\letin{x = z \makevar{\lambda u^\unitt . M_1}{\lambda v^\intt . M_2}}{N} : \theta $}
This is very similar to the previous case: the difference is that the $j_z$ moves from the last case can now be either $read$ or $write(j)$, leading to playing as either $\sem{M_1}$ or $\sem{M_2}$ respectively.  The formal construction is almost identical to that given above.

\section{Proofs of Undecidability Results}\label{appendix:undecidability}
\subsection{Proof of Theorem \ref{thm:rhsundecidability}}
\subsubsection{Q-Stores}

Following previous game semantics based undecidability results, we will reduce the halting problem for a class of finite state machines equipped with a queue to observational equivalence of RML-terms.  The universality of such machines goes back to Post's work on simple rewriting systems~\cite{Pos43,Min67}.  In particular, we will utilise automata equipped with a Q-store~\cite{Mur03}.   Q-stores are a generalisation of a queue which do not always follow queue behaviour.  However, we will be able to detect whether the queue discipline has been followed correctly or not. 

\begin{definition}\rm

A \emph{Q-store}\index{Q-Store} stores characters from a finite alphabet $\Sigma$.  Its content is defined by a natural number $n$ and a function $f : \makeset{0,\ldots,n} \rightarrow \Sigma \times \makeset{+,-} \times \makeset{+,-}$.  The three fields of $f(i)$ will be referred to as $f(i).\mathit{SYMBOL}$, $f(i).\mathit{ACCESSED}$ and $f(i).\mathit{MARKED}$ respectively.  The first holds the character stored in this element of the Q-store and the other two are used for bookkeeping.

The empty Q-store is defined by $n = 0$ and $f(0) = (\dagger,+,-)$ where $\dagger$ is a dummy symbol set as accessed but unmarked.

There are two operations which can be performed on a Q-store.
\begin{itemize}
	\item ADD $x$ adds $x \in \Sigma$ to the store.  The new Q-store $f' : \makeset{0,\ldots,n+1} \rightarrow \Sigma \times \makeset{+,-} \times \makeset{+,-}$ is defined by $f \subseteq f'$, $f'(n + 1) = (x,-,-)$.
	
	\item FETCH is the only access method.  It can return any previously unaccessed element in the store $f(i).\mathit{SYMBOL}$ (i.e.\ $f(i).\mathit{ACCESSED} = -$) provided an index $j$ can be found such that $0 \leq j < i \leq n$, $f(j).\mathit{ACCESSED} = +$ and $f(j).\mathit{MARKED} = -$.  As well as returning the value stored in the $i$th element, the operation sets $f(i).\mathit{ACCESSED}$ and $f(j).\mathit{MARKED}$ to $+$.
	 
\end{itemize}

\end{definition}

We see that a FETCH operation can return any unaccessed element $i$ provided there is an earlier element $j$ which has already been accessed but has not yet been marked.  The choice of $(i,j)$ is made nondeterministically and different choices can affect the store in different ways.  It is possible that the Q-store might behave as a queue.  This will occur if during a FETCH the choice of $i$ will always be the first unaccessed element and $j$ to be $i-1$.  If this happens then the Q-store will have a characteristic pattern: no unaccessed element occurs between two accessed elements.  The only way to have a Q-store with this pattern is if its behaviour has been that of a queue.  In particular, if all elements of a Q-store have been accessed then its behaviour was that of a queue.

We can now consider finite state machines equipped with Q-stores.

\begin{definition}\rm

A \emph{Q-machine}\index{Q-Machine} is a tuple $\mathcal{A} = \anglebra{Q,\Sigma,q_0,F,\delta^{\mathit{ADD}},\delta^{\mathit{FETCH}}}$, where:
\begin{itemize}
	\item $Q = Q^{A} + Q^{F} + F$ is the finite set of states with $q_0 \in Q$ the initial state.
	\item $\delta^{\mathit{ADD}} : Q^{A} \rightarrow Q \times \Sigma$ defines transitions out of states in $Q^A$.  If the machine is in state $q_1$ and $\delta^{\mathit{ADD}}(q_1) = (q_2,a)$ then the machine transitions into state $q_2$ and performs ADD $a$ on the machine's Q-store.
	\item $\delta^{\mathit{FETCH}} : Q^{F} \times \Sigma \rightarrow Q$ defines the machine's action when in a state from $Q^{F}$.  When in state $q_1 \in Q^{F}$ the Q-machine will attempt to perform a FETCH.  If this is  successful and returns symbol $a$ then the machine transitions into state $\delta^{\mathit{FETCH}}(q_1,a)$.  
\end{itemize}

We say that a Q-machine \emph{halts} if there exists a run (starting in the initial state) which ends in a final state (a state in $F$) with a Q-store in which all elements have been accessed.

\end{definition}

Since Q-machines only halt when every element in the Q-store has been accessed (so when the Q-store has acted as a queue) as far as halting is concerned they are the same as finite state automata equipped with a queue.  Hence, from Post's work we can infer that they have an undecidable halting problem.

\subsubsection{Representing Q-machines}

We now consider how to represent the run of an arbitrary Q-machine at the type sequent $\seq{}{ (\unitt \rightarrow \unitt) \rightarrow  \unitt \rightarrow \unitt}$.  The relevant prearena is shown in Figure~\ref{fig:PrearenaQMachine}.
\input{figPreArenaQMachine}
For technical convenience we will assume that the initial state of the Q-store results from a dummy ADD action executed once at the very start of the run.

Our representation of the Q-machine will begin with 
\Pstr{ (q0){q_0} \ (a0-q0,50){a_0} }.

Each ADD operation (including the dummy operation initializing the store) will then be interpreted by the segment
\Pstr{ (q1){q_1} \ (qh-q1,50){\hat{q}} \ (q1a) {q_1} \ (a1-q1a,50) {a_1} }.

Each FETCH will be represented by segments
\Pstr{ (q2) {q_2} \ (qh) {\hat{q}} \ (q2a) {q_2} \ (a2a-q2a,35) {a_2} \ (ah-qh,35) {\hat{a}} \ (a2-q2,45) {a_2}} 
where the first $q_2$ is justified by the $a_1$ from the $i$th ADD, $\hat{q}$ is justified by the $q_1$ immediately before that $a_1$ and the second $q_2$ is justified by the $a_1$ in the $j$th ADD.  Here we are using the visibility condition to force the choice of $j$ to be a strictly earlier ADD-block than the choice of $i$.
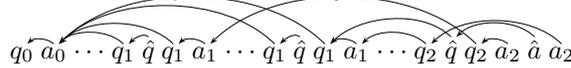
\[
\Pstr{
(q0) {q_0} \ (a0-q0) {a_0} \ 
(d1) {\cdots} \ 
(q1ha-a0) {q_1} \ (qha-q1ha) {\hat{q}} \ (q1a-a0) {q_1} \ (a1a-q1a) {a_1} \ 
(d2) {\cdots} \ 
(q1hb-a0) {q_1} \ (qhb-q1hb) {\hat{q}} \ (q1b-a0) {q_1} \ (a1b-q1b) {a_1} \ 
(d3) {\cdots}  \
(q2b-a1b) {q_2} \ (qhb2-q1b) {\hat{q}} \ (q2a-a1a) {q_2} \ (a2a-q2a) {a_2} \ (ah2-qhb2) {\hat{a}} \ (a2b-q2b) {a_2}
}
\]

Once the Q-machine has reached a final state at the end of the computation, we must check that the Q-store has the correct shape.  This is performed in a finishing up state where we visit each ADD-block from last to first and check each of them has been accessed.

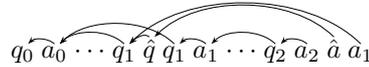
\[
\Pstr{
(q0) {q_0} \ (a0-q0) {a_0} \ 
(d1) {\cdots} \ 
(q1a-a0) {q_1} \ (qh-q1a) {\hat{q}} \ (q1b-a0) {q_1} \ (a1b-q1b) {a_1} \ 
(d2) {\cdots} \ 
(q2-a1b) {q_2} \ (a2-q2) {a_2} \ (ah-qh) {\hat{a}} \ (a1a-q1a) {a_1} 
}
\]

In order to construct a term which follows this strategy we first consider some terms which perform the various responses.  Our final term will keep track of which state the simulation is in and imitate one of these terms accordingly.
\begin{itemize}
	\item $\lambda f . \ldots$ will respond to the initial $q_0$ with $a_0$.
	\item $\lambda f . f(); \lambda x . \Omega$ responds to $q_1$ with $\hat{q}$.  Once this is (eventually) answered with $\hat{a}$ it responds with $a_1$.  This $a_1$ can never be used to justify anything or else P will not respond.
	\item $\lambda f . \lambda x . f()$ responds to $q_1$ with $a_1$.  If this $a_1$ is used to justify a $q_2$ then it responds with $\hat{q}$.  If this is answered with $\hat{a}$ then it responds with $a_2$.
	\item $\lambda f . \lambda x . ()$ responds to $q_1$ with $a_1$ and to $q_2$ with $a_2$.
\end{itemize}

In order to keep track of which stage of the computation we are in, we will use a number of global variables.

\begin{itemize}
	\item \textit{State} --- keeping track of which state the simulated Q-machine is in.
	\item \textit{First} --- a flag letting us know if the first dummy ADD-operation has occurred.
	\item \textit{AddState} --- keeping track of how far through an ADD-operation we are.
	\item \textit{FetchState} --- keeping track of how far through a FETCH-operation we are.
	\item \textit{FinishingState} --- keeping track of how far through a finishing up operation we are.

\end{itemize}

Additionally, we will create several local variables for each ADD.

\begin{itemize}
	\item \textit{Symbol}, \textit{Accessed} and \textit{Marked} --- representing the appropriate fields in the Q-store.
	\item \textit{Finalised} --- a flag keeping track of whether this ADD-operation has been visited during the finishing up stage.  This is needed to ensure that each ADD is visited exactly once during this phase.
\end{itemize}
The term is shown in Figure~\ref{fig:TheTermEncodingAQMachine}.  We use the syntax $[B_1, \ldots, B_n]$ as an abbreviation for $\cond{\bigwedge{B_i}}{()}{\Omega}$.  The local variables are associated with the $q_1 \cdot a_1$ part of each ADD-block.  This ensures they can be accessed during a FETCH or the finishing up stage when moves are hereditarily justified by them.  Note that we cannot enforce that during the finishing up stage, the $q_2$ is justified by the last unfinalised $a_1$.  However, we do ensure that each $a_1$ justifies at most one $q_2$ during this phase.  Since we can rely on the second part of the finishing up state ($\hat{a} \cdot a_1$) to hide (by visibility) the $a_1$ from the last (by bracketing) unfinalised ADD-block, we know that the only way to reach a complete play is if O does indeed finalise the ADD-blocks in order from last to first.

To establish undecidability we note that the represented Q-machine will halt if and only if the term is not observationally equivalent to $\lambda f . \Omega$.  Hence, observational equivalence is undecidable if the type contains a first-order (or higher) argument which is \emph{not} the final argument (i.e.\ any type of the form $\theta_n \rightarrow \ldots \rightarrow \theta_4 \rightarrow (\theta_3 \rightarrow \theta_2) \rightarrow \theta_1 \rightarrow \theta_0$ for any RML types $\theta_i$ and $n \geq 3$).  

\input{figTheTermEncodingAQMachine}

\subsection{Proof of Theorem \ref{thm:recursive-undecidability}}
We again rely on finite state systems equipped with a queue.  However, rather than rely on Q-machines, this time we utilise a programming system called \emph{Queue}.

\begin{definition}\rm

A \emph{Queue program} has a single memory cell $z$ that can store a symbol from $\Sigma$ and a queue (which can contain symbols from $\Sigma$).  A program consists of a finite sequence of instructions of the form $1 : I_1, 2 : I_2, \ldots, m : I_m$, where each $I_i$ is one of the following:

\begin{itemize}
	\item {\normalfont{\texttt{enqueue}} $a$}: add the symbol $a \in \Sigma$ to the end of the queue and go to the next instruction.
	\item {\normalfont\texttt{dequeue}}: if the queue is empty then halt, otherwise remove the element at the front of the queue and store it in $z$ then go to the next instruction.
	\item {\normalfont\texttt{if $z = a$ goto $L$}} where $a \in \Sigma$ and $L \geq 0$ is a label.  If the value stored in $z$ is $a$ then go to the $L$th instruction, otherwise go to the next instruction.
	\item {\normalfont\texttt{halt}}.
\end{itemize}
\end{definition}
The halting problem for Queue programs is undecidable~\cite{JM78}.

We will simulate Queue programs using a recursive function of type $(\unitt \rightarrow \unitt) \rightarrow \unitt$.  We will model the queue using the call-stack.  Every \texttt{enqueue} will cause a recursive call which will allocate a variable $\mathit{cur}$ containing the value to be enqueued.  When an item is removed from the queue we will set $\mathit{cur}$ to $0$ which we assume is a special value not in $\Sigma$.  This means that we know that the head of the queue corresponds to the oldest recursive call whose $\mathit{cur}$ does not contain $0$.  

In addition to the local variable $\mathit{cur}$ we will also need global variables $\mathit{halt}$ (a flag letting us know we should stop the computation and collapse the call-stack), $\mathit{pc}$ (which instruction we are currently on), $z$ (the Queue program's memory cell) and two variables $G$ and $H$.  When we make our recursive call, the new value to be added to the queue will be (temporarily) stored in $G$.  Further, the argument to the call (a function of type $\unitt \rightarrow \unitt$) will be such that if it is run when $H=0$ then the value of $\mathit{cur}$ from the previous call will be written to $G$.  If, on the other hand, the argument is run when $H = 1$ it will cause the value at the front of the queue to be written to $G$ and the appropriate $\mathit{cur}$ to be set to $0$ (i.e.\ that element is removed from the queue).  

Our term encoding a queue program is then
\[\begin{array}{l}
	\letin{\mathit{halt},\mathit{pc},z,G,H=\refint{0},\refint{1},\refint{0},\refint{0},\refint{0}}{}\\
	\qquad \qquad \qquad(\mu F^{(\unitt \rightarrow \unitt) \rightarrow \unitt}. \lambda {\mathit{arg}}^{\unitt \rightarrow \unitt}. {\mathit{body}}) (\lambda c^{\unitt} . \Omega)
\end{array}\]
where $\mathit{body}$ has the form
\[\letin{\mathit{cur} = \refint{(\deref{G})}}{\while{\deref{\mathit{halt}} = 0}{\rmlterm{case}(\deref{\mathit{pc}})[1 \mapsto J_1, \ldots, m \mapsto J_m]}}.\]

Each $J_i$ depends on $I_i$ according to Table~\ref{tab:SimulationsForEachQueueProgramInstruction}.
\input{tabSimulationsForEachQueueProgramInstruction}
This term is equivalent to $\seq{}{()}$ if and only if the simulated Queue program halts.  Hence, observational equivalence of $\rmlostr$ with recursive functions of type $(\unitt \rightarrow \unitt) \rightarrow \unitt$ is undecidable.  

\end{document}

%% file: macros.tex
\usepackage{amsmath,amssymb,amsfonts,mathdots}
\usepackage{stmaryrd}
\usepackage{color}

\usepackage{listings}
\usepackage{enumerate}

\usepackage{mathpartir}
\usepackage{caption} 

\usepackage{subcaption}
\usepackage{pstring}
\pstrSetArrowColor{black}

\newcommand\anglebra[1]{\left< \,#1 \, \right>}

\newcommand\rmlostr{\hbox{\rm RML}_{\hbox{\rm \scriptsize O-Str}}}
\newcommand\rmlpstr{\hbox{\rm RML}^{\hbox{\rm \scriptsize P-Str}}_{\seq{2}{1}}}

\newcommand\rml{{\hbox{\rm RML}}}
\newcommand\forml{\hbox{\rm RML}_{\seq{2}{1}}}
\newcommand\rforml{\hbox{\rm RML}_{\seq{2}{1}}^\mathrm{res}}
\newcommand\frml{\rml_{\rm f}}

\newcommand\rmlterm[1]{{\bf #1}}
\newcommand\rmltype[1]{{\sf #1}}

\newcommand\intt{\rmltype{int}}
\newcommand\unitt{\rmltype{unit}}
\newcommand\intreft{\rmltype{int\ ref}}

\newcommand\suc[1]{\rmlterm{succ}(#1)}
\newcommand\pre[1]{\rmlterm{pred}(#1)}
\newcommand\letin[2]{\rmlterm{let}\,#1\,\rmlterm{in}\,#2}

\newcommand\comp[1]{\rmlterm{comp}(#1)}

\newcommand\makevar[2]{\rmlterm{mkvar}(#1,#2)}

\newcommand\refint[1]{\rmlterm{ref}\,#1}
\newcommand\newc{\refint{0}}

\newcommand{\aasg}{\,\raisebox{0.065ex}{:}{=}\,}
\newcommand\assg[2]{#1\aasg#2}
\newcommand\deref[1]{!#1}
\newcommand\while[2]{\rmlterm{while}\,#1\,\rmlterm{do}\,#2}
\newcommand\cond[3]{\rmlterm{if}\,#1\,\rmlterm{then}\,#2\,\rmlterm{else}\,#3}

\newcommand\seq[2]{{#1} \vdash {#2}}
\newcommand\sem[1]{{\llbracket}{#1}{\rrbracket}}
\newcommand\sseq[2]{\sem{\seq{#1}{#2}}}

\newcommand\can{\mathbb{C}}

\newcommand\ptarget[1]{\stackrel{\bullet}{#1}}
\newcommand\psource[1]{\stackrel{\circ}{#1}}

\newcommand\twovec[2]{\Bigl(\negthinspace\begin{smallmatrix}#1\\#2\end{smallmatrix}\Bigr)}
\newcommand\threevec[3]{\Bigl(\negthinspace\begin{smallmatrix}#1\\#2\\#3\end{smallmatrix}\Bigr)}

\newcommand\pred{\mathit{pred}}
\newcommand\set[1]{\{#1\}}
\newcommand\powerset{{\cal P}}
\newcommand\makeset[1]{\set{#1}}

\newcommand\calA{{\cal A}}

\newcommand\calC{{\cal C}}
\newcommand\calD{{\cal D}}
\newcommand\dataset{\calD}
\def\D{\mathbb{D}}

\newcommand\intnum{\mathbb{Z}}
\newcommand\natnum{\mathbb{N}}

\usepackage{tikz}
\usetikzlibrary{matrix}
\usetikzlibrary{chains}
\usetikzlibrary{arrows}
\usetikzlibrary{shapes}
\usetikzlibrary{positioning}
\usetikzlibrary{calc}
\usetikzlibrary{fit}
\usetikzlibrary{automata}

\tikzstyle{play} = [start chain, 
every node/.style={draw=none,fill=none,on chain},
 node distance=-0.5ex,
 text depth=-.755ex]

\tikzstyle{tri} = [draw,isosceles triangle, anchor=apex]
\tikzstyle{clear} = [draw=none,fill=none]
\tikzstyle{arena} = [node distance=1em,every node/.style={clear}]
\tikzset{initial text={}}

%% file: figPreArenaQMachine.tex
\begin{figure}

\begin{center}
\begin{tikzpicture}[arena]
\node(q0) {$q_0$};
\node(a0)[below=of q0] {$a_0$} edge (q0);
\node (q1)[below = of a0] {$q_1$} edge (a0);
\node (a1)[below = of q1] {$a_1$} edge (q1);
\node (q2)[below = of a1] {$q_2$} edge (a1);
\node (a2)[below = of q2] {$a_2$} edge (q2);
\node (qhat)[left = of a1] {$\hat{q}$} edge (q1);
\node (ahat)[below = of qhat]{$\hat{a}$} edge (qhat);

\end{tikzpicture}
\end{center}

	\caption{Prearena for $\seq{}{ (\unitt \rightarrow \unitt) \rightarrow  \unitt \rightarrow \unitt}$}
	\label{fig:PrearenaQMachine}
\end{figure}

%% file: figTheTermEncodingAQMachine.tex
\begin{figure}
	\begin{lstlisting}[basicstyle=\small,tabsize=2]
let
	State = ref $q_0$
	First = ref 1
	AddState = ref 0
	FetchState = ref 0
	FinishingState = ref 0
in
	$\lambda$ f .
		[!State $\in Q^{A}$];
		if !AddState = 0 then
			AddState := 1;
			f();
			[!State $\in F$, !FinishingState = 1];
			FinishingState := 0;
			$\lambda$ x . $\Omega$
		else if !AddState = 1 then
			let 
				Symbol = ref $\ddag$
				Accessed = ref (if !First then + else -)
				Marked = ref -
				Finalised = ref -
			in 
				AddState := 0;
				if !First then
					First := 0;	Symbol := $\dagger$;
				else
					(Symbol,State) := $\delta^{\mathit{ADD}}$(!State);
				$\lambda$ x .
					if !State $\in Q^{F}$ then
						if !FetchState = 0 then
							[!Accessed = -];
							Accessed := +; FetchState := 1;
							f();
							[!FetchState = 2];
							FetchState := 0; 
							State := $\delta^{\mathit{FETCH}}$(!State,!Symbol);
						else if !FetchState = 1 then
							[!Accessed = +, !Marked = -];
							FetchState := 2; Marked := +;
						else $\Omega$
					else if !State $\in F$ then
						[!FinishingState = 0, !Accessed = +, !Finalised = -];
						FinishingState := 1; Finalised := +;
					else $\Omega$
		else $\Omega$	
\end{lstlisting}
	\caption{The term encoding a Q-machine}
	\label{fig:TheTermEncodingAQMachine}
\end{figure}

%% file: tabSimulationsForEachQueueProgramInstruction.tex
\begin{table}
	\centering
		\begin{tabular}{|l|l|}
			\hline
\multicolumn{1}{|c|}{$I_i$}
 &\multicolumn{1}{|c|}{$J_i$}\\
			\hline
			\texttt{enqueue} $n$ &  
				$\begin{array}{l}\assg{pc}{i+1};\\
				\assg{G}{n};\\
				F(\lambda x . \cond{!H = 0}{L}{R})\\
				\\
				\textrm{where}\\
				L\equiv\assg{G}{\deref{\mathit{cur}}}\\
				R\equiv	\rmlterm{if}\,{(\assg{H}{0};\mathit{arg}();!G =0)}\,\rmlterm{then}\,{\assg{z}{\mathit{cur}};\assg{\mathit{cur}}{0}}\,\\\qquad\rmlterm{else}\,{\assg{H}{1};\mathit{arg}()}
				\end{array}$\\
				\hline
				
				\texttt{dequeue} &
				$\begin{array}{l}
				\cond{\deref{\mathit{cur}}=0}{\assg{\mathit{halt}}{1}}{}\\
				\qquad \cond{\assg{H}{0}; \mathit{arg}(); \deref{G}=0}{\assg{z}{\deref{\mathit{cur}}};\assg{\mathit{cur}}{0}}\\
				\qquad \qquad{\assg{H}{1};\mathit{arg}()};\\
				\qquad \assg{\mathit{pc}}{i + 1}\\
				
				\end{array}$\\
				\hline
				\texttt{halt} & $\assg{\mathit{halt}}{1}$\\
				\hline
				\texttt{if $z = n$ goto $L$} & $\cond{\deref{z} = n}{\assg{\mathit{pc}}{L}}{\assg{\mathit{pc}}{i+1}}$\\
				\hline
		\end{tabular}
	\caption{Simulations for each Queue program instruction}
	\label{tab:SimulationsForEachQueueProgramInstruction}
\end{table}